\documentclass[a4paper]{amsart}
\usepackage{amsmath}
\usepackage{amssymb}
\usepackage{amsthm}
\usepackage{amsfonts}
\usepackage{amstext}
\usepackage{amsopn}
\usepackage{amsxtra}
\usepackage{graphicx}
\usepackage{pstricks}
\usepackage[latin1]{inputenc}
\usepackage{dsfont}
\usepackage{mathrsfs}
\usepackage{bm}
\usepackage{pdfsync}

\newtheorem{lemma}{Lemma}[section]
\newtheorem{theorem}{Theorem}

\theoremstyle{definition}

\theoremstyle{definition}
\newtheorem{remark}[lemma]{Remark}
\theoremstyle{definition}

\newcommand{\gH}{\mathfrak{H}}
\newcommand{\gS}{\mathfrak{S}}

\newcommand{\X}{\mathcal{X}}

\newcommand{\cE}{\mathcal{E}}
\newcommand{\cB}{\mathcal{B}}

\newcommand{\cD}{\mathcal{D}}

\newcommand{\cF}{\mathcal{F}}
\newcommand{\D}{\mathcal{D}}

\newcommand{\cC}{\mathcal{C}}

\newcommand{\eps}{\epsilon}
\newcommand{\C}{\mathbb{C}}
\newcommand{\R}{\mathbb{R}}

\newcommand{\sig}{\boldsymbol{\sigma}}

\newcommand\ii{{\ensuremath {\infty}}}

\newcommand\1{{\ensuremath {\mathds 1} }}
\newcommand{\wto}{\rightharpoonup}
\newcommand\pscal[1]{{\ensuremath{\left\langle #1 \right\rangle}}}
\newcommand{\norm}[1]{ \left| \! \left| #1 \right| \! \right| }
\newcommand{\CJ}{\mathscr{C}}
\renewcommand{\phi}{\varphi}
\def\tr{\mathop{\rm tr}\nolimits} 
\def\Tr{{\rm Tr}_{\C^2}}
\newcommand\vF{v_{\rm F}}
\newcommand\veff{v_{\rm eff}}

\begin{document}

\title{Ground state properties of graphene in Hartree-Fock theory}

\author[C. Hainzl]{Christian Hainzl}
\address{Mathematisches Institut, Auf der Morgenstelle 10, D-72076 Tübingen, Germany}
\email{christian.hainzl@uni-tuebingen.de }

\author[M. Lewin]{Mathieu Lewin}
\address{CNRS \& University of Cergy-Pontoise, Mathematics Department (UMR 8088), F-95000 Cergy-Pontoise, France} 
\email{mathieu.lewin@math.cnrs.fr}

\author[C. Sparber]{Christof Sparber}
\address{Department of Mathematics, Statistics, and Computer Science, University of Illinois at Chicago, 
851 South Morgan Street, Chicago, IL 60607, USA}
\email{sparber@math.uic.edu}

\begin{abstract}
We study the Hartree-Fock approximation of graphene in infinite volume, with instantaneous Coulomb interactions. First we construct its translation-invariant ground state and we recover the well-known fact that, due to the exchange term, the effective Fermi velocity is logarithmically divergent at zero momentum. In a second step we prove the existence of a ground state in the presence of local defects and we discuss some properties of the linear response to an external electric field. All our results are non perturbative.

\bigskip

\noindent\tiny\copyright 2012, by the authors. This paper may be reproduced, in its entirety, for non-commercial purposes.
\end{abstract}

\maketitle

\begin{center}{\it Dedicated to Elliott H. Lieb on the occasion of his 80th birthday}\end{center}

\tableofcontents

\section{Introduction}

Graphene is a mono-crystalline
graphitic film in which electrons behave like 2-dimensional Dirac fermions without mass. It has attracted a huge interest in the last decades~\cite{NetGuiPerNovGei-09}.
The Fermi surface of metals in dimension $d$ is usually (the union of) $(d-1)$-dimensional manifold(s), whereas for gapless semi-conductors it is often composed of parabolic-like points. Graphene, on the other hand, has a peculiar quasi-particle dispersion relation which is conical at the degeneracy points, leading to the effective massless Dirac equation. This is not so exceptional, however. It was recently shown that conical singularities are generic in (non-relativistic) quantum crystals having the honeycomb lattice symmetry~\cite{FefWei-12}. The conical dispersion relation, combined with the fact that the Fermi velocity $\vF$ is 300 times smaller than the speed of light, makes graphene an ideal condensed matter system for testing in the lab our understanding of 2D massless relativistic particles~\cite{Novoselovetal-05,Zhouetal-06}.

Quantum Electrodynamics (QED) is a very powerful theory which, however, has only been rigorously formulated in a perturbative fashion. 
The small value of the Fermi velocity in graphene restricts the validity of perturbation theory and it is important to be able to resort to non-perturbative methods. Short range interactions in graphene are well understood~\cite{GonGuiVoz-01,GiuMas-09,GiuMas-10}. They do not modify the general properties of the ground state as compared to the free case. 
The effect of long range Coulomb interactions, however, is much less clear. It was argued in~\cite{GonGuiVoz-99,Mishchenko-07} that instantaneous Coulomb interactions lead to an additional logarithmic divergence in the infrared regime for the interacting Fermi velocity, whereas the authors of~\cite{GonGuiVoz-94} claimed that retarded effects make it tend to the speed of light. In~\cite{GiuMasPor-10} a careful renormalization group analysis for retarded interactions indicated that the Fermi velocity actually has a limit, but that it is in general smaller than the speed of light.

In this work we shall not consider retardation effects, but instead confirm the logarithmic divergence in the case of instantaneous Coulomb interactions, in a \emph{fully non-perturbative setting}, based on a Hartree-Fock model.
In~\cite{HaiLewSer-05a,HaiLewSer-05b,HaiLewSol-07}, two of the authors of the present paper have, with \'E. Séré and J.P. Solovej, studied the Hartree-Fock approximation of 3-dimensional Quantum Electrodynamics. These works were inspired by an important paper of Chaix and Iracane~\cite{ChaIra-89} and they have been summarized in~\cite{HaiLewSerSol-07,EstLewSer-08}. They are purely non-perturbative and hold for all values of the coupling constant $0\leq \alpha < 4/\pi\simeq 1.27$. For $\alpha>4/\pi$, the system is known to become unstable~\cite{ChaIraLio-89}. In the present work we extend these results to (massless) electrons in two dimensions. 
We identify the exact Hartree-Fock ground state of the system at half filling and zero temperature under the sole assumption that 
$$0\leq\alpha\leq 0.48637 \quad\text{or, equivalently,}\quad \vF\geq 2.0560,$$ 
which covers the case of graphene (see Section \ref{sec:HFg} below).

Our methodology is as follows. We consider a Hartree-Fock type model in which particles interact through the instantaneous Coulomb potential and with a kinetic energy given by the massless Dirac operator. Since we do not use normal-ordering, the Hamiltonian is unbounded from below. However, we shall, as a first step, construct the free Dirac sea (i.e., the absolute minimizer of the energy in the absence of external fields) by means of a thermodynamic limit. This state corresponds to filling the negative energies of an effective mean-field translation-invariant operator of the form $\veff(p) \sig\cdot p$. Here, $\veff(p)$ denotes an effective Fermi velocity which we shall compute exactly and show that,
for small momentum, it diverges like
$$\veff(p)\underset{p\to0}\sim \frac{1}{4}\log\frac{\Lambda}{|p|},$$ where $\Lambda$ is a fixed ultraviolet cut-off. 
In a second step, we shall introduce an external electrostatic field and obtain a bounded-below energy by subtracting the (infinite) energy of the free Dirac sea. 
This enables us to prove the existence of a ground state in the presence of the external field, in infinite volume. In other words, we use the translation-invariant free state as a reference and we describe variations compared to it. 

In addition, we shall calculate the linear response function $B(p)$ of the density of charge of graphene to an external electrostatic field, and we shall show that, in the limit of small momenta, 
$$B(p) \underset{p\to0}{\sim} \frac\pi{4\,\log\left(\frac{\Lambda}{|p|}\right)}.$$ 
It is interesting to note that the dielectric behavior of Hartree-Fock graphene is universal at low momentum, that is, independent of the bare Fermi velocity $\vF$. As we shall see, this follows from the fact that the effective Fermi velocity $\veff(p)$ has the same property.

As compared to the previously quoted results in 3 dimensions, we deal here with \emph{massless} Dirac fermions. Mathematically, this creates a certain lack of control of the behavior of minimizing sequences in the infrared domain, which complicates the study of the exchange term. Interestingly, though, our existence proof for ground states heavily relies on the effective Dirac operator including the logarithmically effective Fermi velocity, and which can be used to get a better control. On the other hand, the ultraviolet behavior which was problematic in the three dimensional case is not an issue here. For graphene it would even have no meaning since the Dirac dispersion relation is only valid for low momentum anyhow. In this paper we consequently impose a sharp ultraviolet cut-off $\Lambda$ which is kept fixed all along our study and which mimics the presence of the underlying lattice.

\medskip

The paper is organized as follows. In the next section we properly introduce the Hartree-Fock approximation of massless 2D QED. Then, in Section~\ref{sec:free} we construct the free Fermi sea of graphene and we discuss the Fermi velocity divergence at low momentum due to the Coulomb exchange term. In Section~\ref{sec:BDF} we prove the existence of a ground state in the presence of external fields, like those induced by defects in graphene. This will allow us to compute the linear response of graphene in Section~\ref{sec:linear_response}. Section~\ref{sec:proof} contains the proof of our main theorem, whereas in the Appendix we prove a useful localization formula for the massless pseudo-relativistic 
kinetic operator $\sqrt{-\Delta}$.

\bigskip

\noindent\textbf{Acknowledgment.} M.L. acknowledges financial support from the French Ministry of Research (ANR-10-BLAN-0101) and from the  European Research Council under the European Community's Seventh Framework Programme (FP7/2007-2013 Grant Agreement MNIQS 258023).

\section{Hartree-Fock theory of graphene}\label{sec:HFg}

In this section we explain how to derive our model from 2D Quantum Electrodynamics, following~\cite{HaiLewSol-07}. We consider a system of 2D 
massless electrons interacting through the 3D instantaneous Coulomb potential. The latter is justified by the fact that while the electrons in graphene are essentially confined in 2D, the electric field clearly still acts in all three 
spatial dimensions.
The corresponding \emph{formal} 2D QED-type Hamiltonian, written in Coulomb gauge, then reads \cite{GonGuiVoz-99,GonGuiVoz-01,BorPolAsgMac-09,NetGuiPerNovGei-09}
\begin{multline}
\mathbb{H}^{V}= \vF\int_{\R^2} \Psi^*(x)\, \sig \cdot(-i\nabla)\Psi(x) \,dx+\int_{\R^2} V(x)\rho(x)\,dx \\
+ \frac{1}2 \iint_{\R^2\times\R^2} \frac{\rho(x)\rho(y)}{|x-y|}dx\,dy
\label{Ham}
\end{multline}
with $\sig = (\sigma^1, \sigma^2)$, the first two Pauli-matrices.  Here, and in the following, we shall use the notation
$$D^0 = -i \sig \cdot \nabla = -i\sigma^1 \partial_{x_1}  -i\sigma^2 \partial_{x_2}$$
for the massless 2D Dirac operator.

We are working in atomic units with the mass and charge of the electrons normalized to 1. The (bare) Fermi velocity $\vF$ is about $\vF\simeq 2.20$ in graphene. Thus the electrostatic interaction, i.e., 
the last term of~\eqref{Ham}, and the kinetic part are of the same order. The internal variable of the electrons is a \emph{pseudo-spin} which describes to which of the two interlaced triangular lattices of carbon which make up the hexagonal lattice the electron belongs. Since we will not consider magnetic fields in this paper, we have neglected the (regular) spin of the electron for simplicity. The previous Hamiltonian describes well the electrons with energies about the Fermi energy of graphene, provided that they live on a much larger scale than the lattice size. In this regime graphene can be seen as a continuous medium, leading to the 2D QED-type Hamiltonian~\eqref{Ham}.

Let us now briefly explain the main objects which enter in the definition of the Hamiltonian~\eqref{Ham}. In there $\Psi(x)$ is the second quantized field operator which annihilates an electron at $x$ and satisfies the anti-commutation relation
\begin{equation}
 \Psi^*(x)_\sigma\Psi(y)_\nu+\Psi(y)_\nu\Psi^*(x)_\sigma = 2\delta_{\sigma,\nu}\delta(x-y),
\label{CAR}
\end{equation}
with $\sigma,\nu\in\{\pm1/2\}$ the pseudo-spin variables.
The operator $\rho(x)$ is the \emph{density operator} defined by
\begin{equation}
\rho(x) =  \sum_{\sigma
=1}^2\frac{[\Psi_\sigma^*(x),\Psi_\sigma(x)]}{2},
\label{def_rho}
\end{equation}
where $[a,b]=ab-ba$. Finally, the function $V$ is a local external electrostatic potential which is applied to the system. It is for instance induced by a set of defects in the system. 
The Hamiltonian $\mathbb{H}^{V}$ now formally acts on the fermionic Fock space $\mathscr F$ for the electrons. The commutator in the definition~\eqref{def_rho} of $\rho(x)$ is important as it ensures charge conjugation invariance~\cite{Schwinger-48}. Precisely, we have 
$$\CJ\rho(x)\CJ^{-1}=-\rho(x),\qquad \CJ\mathbb{H}^{V} \CJ^{-1}=\mathbb{H}^{-V},$$
where $\CJ$ is the charge conjugation operator acting on the Fock space. The Hamiltonian $\mathbb{H}^{V}$ is  unbounded from below on $\mathscr F$ and it is not even a well-defined self-adjoint operator. 
However, it is still possible to define it in a box with suitable boundary conditions and with an ultraviolet cut-off, as was done in~\cite{HaiLewSol-07} (see Section \ref{sec:free} below for more details).

We have already neglected photons in our model. We shall now make another approximation, by restricting our attention to Hartree-Fock states. Let us recall that the electronic \emph{one-body density matrix} (two point function) of any electronic state $\Omega$ is defined as
$$\gamma(x,y)_{\sigma,\sigma'}=\langle\Psi^*(x)_\sigma\Psi(y)_{\sigma'}\rangle_\Omega.$$
It is an operator on the one-body space such that $0\leq \gamma\leq 1$, due to the anti-commutation relations. It is an orthogonal projection for (and only for) pure Hartree-Fock states.
In view of \eqref{def_rho}, it is natural to introduce a \emph{renormalized one-body density matrix}
$$\gamma_{\rm ren}(x,y)_{\sigma,\sigma'}=\pscal{\Omega\left|\frac{[\Psi(x)_{\sigma}^*,\Psi(y)_{\sigma'}]}{2}\right|\Omega}.$$
By \eqref{CAR}, we obtain the simple relation
$$\gamma_{\rm ren}=\gamma-\frac{I}{2},$$
where $I$ is the identity operator on the one-body space. 

Electronic Hartree-Fock states form a subset of states which are completely determined by their density matrix $\gamma$ (or equivalently by their renormalized density matrix $\gamma_{\rm ren}=\gamma-I/2$). 
The energy of any such Hartree-Fock state reads
$$\langle\mathbb{H}^{V}\rangle=\cE_{\rm HF}^V(\gamma-I/2)+\text{C}$$
where the constant $\text{C}$ diverges for infinite volume, and
\begin{multline}
 \cE_{\rm HF}^V(\gamma_{\rm ren})=\vF\tr(D^0\gamma_{\rm ren})+\int_{\R^2} V(x)\rho_{\gamma_{\rm ren}}(x)\,dx\\
+\frac{1}{2}\iint_{\R^2\times\R^2}\frac{\rho_{\gamma_{\rm ren}}(x)\rho_{\gamma_{\rm ren}}(y)}{|x-y|}dx\,dy
-\frac{1}{2}\iint_{\R^2\times\R^2}\frac{|\gamma_{\rm ren}(x,y)|^2}{|x-y|}dx\,dy.
\label{HF_QED}
\end{multline}
The reader can recognize the well-known Hartree-Fock energy~\cite{LieSim-77}, but applied to the renormalized density matrix $\gamma_{\rm ren}=\gamma-I/2$ instead of the usual density matrix $\gamma$. This is of course a consequence of our choice of a charge-conjugation invariant density operator $\rho(x)$ in~\eqref{def_rho}. The last two terms of the first line of~\eqref{HF_QED} are, respectively, the kinetic energy and the interaction energy of the electrons with the external potential $V$. In the second line appear, respectively, the so-called \emph{direct} and \emph{exchange} terms. In Relativistic Density Functional Theory \cite{Engel-02}, the latter is approximated by a function of $\rho_\gamma$ and its derivatives only, a procedure which we shall not follow here.

One of the main goals of this paper is to construct ground states for the Hartree-Fock energy~\eqref{HF_QED} describing electrons in graphene. This energy is not bounded from below and it is not well defined as such, because these states always have infinitely many electrons. But following~\cite{HaiLewSer-05a,HaiLewSer-05b,HaiLewSol-07} we shall see in the next section that it is possible to construct ground states by using a thermodynamic limit procedure. 

In this section we have considered generalized (mixed) Hartree-Fock states, whose density matrix $\gamma$ only fulfills the condition $0\leq\gamma\leq I$. This technique was first proposed by Lieb~\cite{Lieb-81} and it is very convenient when proving existence results for ground states. In view of a variational principle from~\cite{Lieb-81}, Hartree-Fock ground states are always pure in the presence of repulsive interactions, i.e. their density matrix is automatically a projection in the end.

\section{Fermi velocity enhancement in Hartree-Fock graphene}\label{sec:free}

In this section we consider a graphene sheet without any external field, $V\equiv0$, and we investigate the effect of the Coulomb interactions among the electrons. In mean-field theory it is well-known that the effective Fermi dispersion relation becomes singular at $0$. This \emph{enhancement of the Fermi velocity} has already been remarked in~\cite{GonGuiVoz-99,Katsnelson-06,Vafek-07,BorPolAsgMac-09}. Our main contribution in this section is the rigorous proof that the so-obtained state is actually the true ground state of the system. The method thereby follows that of~\cite{HaiLewSol-07}.

If the electrostatic interactions between particles are neglected and the system is confined to a box with periodic boundary conditions, it is obvious that the unique minimizer is the (non-interacting) free Dirac sea, which converges in the thermodynamic limit to the infinite-volume (non-interacting) free Dirac sea. The latter is an infinite Hartree-Fock state containing all the negative energy electrons (in accordance with the old Dirac picture~\cite{Dirac-33,Dirac-34b}), whose density matrix and renormalized density matrix are, respectively, given by
\begin{equation}
P^0_-=\1(D^0\leq0)\qquad\text{and}\qquad 
\gamma^0_{\rm ren}=-\frac{D^0}{2|D^0|}=P^0_--\frac{I}{2}.
\label{eq:def_free_vaccum} 
\end{equation}
When interactions are taken into account, the free Dirac sea changes but it stills remains translation invariant. 
The latter was rigorously proved in the (massive) 3D case in~\cite{HaiLewSol-07}, for $\alpha<4/\pi$. The same result will be true here.

The energy per unit volume of a translation-invariant state $\gamma_{\rm ren}=f_{\rm ren}(p)$ is given by
\begin{equation}
\cF(\gamma_{\rm ren}) = \frac 1{(2\pi)^2} \left( \vF\int_{B(0,\Lambda)} \Tr \big(\sig\cdot p\, f_{\rm ren}(p)\big) - \frac {1}2 \int_{\R^2} \frac{|\check{f}_{\rm ren}(x)|^2}{|x|} dx  \right).
\label{eq:def_energy_per_unit_volume}
\end{equation} 
This energy is bounded from below provided an ultraviolet cut-off $\Lambda$ is inserted. 
In the context of graphene, the ultraviolet cut-off $\Lambda$ mimics the presence of the carbon lattice in graphene. Its physical value is $\Lambda\simeq 0.1\, {\rm \AA}^{-1}$, see~\cite{Zhouetal-06}.
The state $\gamma_{\rm ren}$ is a multiplication operator by the $2\times2$ matrix $f_{\rm ren}(p)$ in Fourier space, supported in the ball of radius $\Lambda$. Its kernel in $x$ space is then given by $(-2\pi)^{-1}\check{f}_{\rm ren}(x-y)$ where $\check{f}_{\rm ren}$ is the Fourier inverse of $f_{\rm ren}$. The density of charge of any such translation-invariant state is found to be constant in space:
$$\rho_{\gamma_{\rm ren}}=(2\pi)^{-1}\tr_{\C^2}\check{f}_{\rm ren}(0)=(2\pi)^{-2}\int_{B(0,\Lambda)}\tr_{\C^2}\big(f_{\rm ren}(p)\big)\,dp.$$
The second term in the energy~\eqref{eq:def_energy_per_unit_volume} is the exchange term per unit volume. We have assumed that our translation-invariant state $\gamma_{\rm ren}$ has no density of charge and thus 
there is no direct term in the energy. We will verify below that, indeed, $\rho_{\gamma_{\rm ren}}\equiv0$ for the minimizer. On physical ground it is clear why this must hold since the Coulomb energy of a constant density of charge is not proportional to the volume unless it vanishes identically. 
Recalling that $\gamma_{\rm ren}=\gamma-I/2$, the constraint $0\leq\gamma\leq I$ then takes the form
\begin{equation}
-\frac{I_{\C^2}}{2}\leq f_{\rm ren}(p)\leq \frac{I_{\C^2}}{2} \qquad\text{for a.e. $|p|\leq \Lambda$,}
\label{eq:constraint_translation_inv}
\end{equation}
where $I_{\C^2}$ is the $2\times2$ identity matrix. 
\begin{remark}
Let us remark that adding the ultraviolet cut-off $\Lambda$ is equivalent to replacing the one-particle Hilbert space $L^2(\R^2,\C^2)$ by the Hilbert space 
\begin{equation}\label{eq:HScut}
\gH_\Lambda:=\{\phi\in L^2(\R^2,\C^2)\ :\ {\rm supp}(\widehat{\phi})\subset B(0,\Lambda)\}.
\end{equation}
\end{remark}
In this section we will show that the \emph{non-interacting} state $\gamma^0_{\rm ren}$, defined in~\eqref{eq:def_free_vaccum} and 
which consists in filling all the negative energies of the free Dirac operator, is the unique ground state of the \emph{interacting} energy per unit volume~\eqref{eq:def_energy_per_unit_volume}. This surprising fact only occurs because of the absence of a mass. It is not true for massive particles, for which the interacting ground state depends in a nonlinear manner on interactions~\cite{LieSie-00,HaiLewSol-07}.

The renormalized density matrix $\gamma^0_{\rm ren}$ is the multiplication operator in the Fourier domain by the matrix $$f^0_{\rm ren}(p)=-\frac{\sig\cdot p}{2|p|}=\1_{(-\ii,0)}(\sig\cdot p)-\frac{I_{\C^2}}{2}.$$
Because the Pauli matrices are trace-less, the charge density of this state vanishes, i.e.,
$\rho_{\gamma^0_{\rm ren}}\equiv0,$
as was announced before. The mean-field (Fock) operator of this state is given by
\begin{equation}
\cD^0=\vF D^0-\frac{\check{f}^0_{\rm ren}(x-y)}{2\pi|x-y|}.
\label{def:cD_0}
\end{equation}
It is nothing else but the derivative of the energy~\eqref{eq:def_energy_per_unit_volume} at $f^0_{\rm ren}$. The second term on the right is the exchange term. The following gives the formula of $\cD^0$ in Fourier space.

\begin{lemma}[Effective velocity of graphene]\label{lem:cD_0}
With $f^0_{\rm ren}(p)=-\sig\cdot p/(2|p|)$, the mean-field translation-invariant operator~\eqref{def:cD_0} can be written as
\begin{equation}
\cD^0(p)=v_{\rm eff}(p)\, \sig\cdot p,\qquad\text{where}\quad 
v_{\rm eff}(p):=\vF+g\left(\frac{\Lambda}{|p|}\right)
\label{eq:cD_0}
\end{equation}
and
\begin{equation}
g(R)=\frac{1}{2\pi}\int_0^{\pi} \int_0^{R} \frac{\cos \theta}{\sqrt{r^2 - 2r \cos \theta + 1 }} r dr d\theta.
\label{eq:g}
\end{equation}
The function $g$ is increasing on $[1,\ii)$. It satisfies 
$$g(1)=\frac{2{\rm G}-1}{2\pi}\simeq 0.1324$$ 
where ${\rm G}=\sum_{n\geq0}(-1)^n(2n+1)^{-2}\simeq0.9160$ is Catalan's constant and 
$$g(r)= \frac{\log(r)}4+O(1)_{r\to\ii}.$$
\end{lemma}

We see that the effective Fermi velocity $v_{\rm eff}(p)$ is logarithmically divergent at $p=0$, i.e.
\[
v_{\rm eff}(p)=\frac{1}{4}\log\frac{\Lambda}{|p|}+O(1)_{|p|\to0}.
\]
This is the well-known velocity enhancement mentioned in the title of the section (for comparison see, e.g.,~\cite[Eq. (220)]{NetGuiPerNovGei-09}). Here the $O(1)$ 
is independent of $\Lambda$. 
\begin{remark}
The logarithmic divergence is sometimes called the \emph{Kohn anomaly}. It has the effect of reducing the density of states near the Dirac energy \cite{NetGuiPerNovGei-09}. 
\end{remark}
Using that $g(\Lambda/|p|)\geq g(1)$, we see that
$$|\cD^0(p)|\geq (\vF+g(1))|D^0(p)|,$$
an inequality that will play an important role later when we will show that $\gamma^0_{\rm ren}$ is the unique minimizer in the absence of external potentials. But before we shall state the proof of Lemma~\ref{lem:cD_0}.

\begin{proof}
Using that the Fourier transform of $|x|^{-1}$ is exactly $|k|^{-1}$ in 2D, we can write the translation-invariant operator $\cD^0$ defined in~\eqref{def:cD_0} in Fourier space as
\begin{equation}
\vF\, \sig\cdot p-\cD^0(p)= \frac{1}{2\pi}\int_{|k|\leq \Lambda}\frac{f^0_{\rm ren}(k)}{|p-k|}\,dk= -\frac{1}{4\pi}|p|\sig\cdot \int_{|k|\leq \Lambda/|p|}\frac{\omega_k }{|k-\omega_p|}\,dk
\label{def:cD_0_Fourier}
\end{equation}
with $\omega_k:=k/|k|$. It is clear that the vector 
$$\int_{|k|\leq \Lambda/|p|}\frac{\omega_k }{|k-\omega_p|}\,dk$$
is co-linear to $p$. Hence we can also write
\begin{equation*}
\int_{|k|\leq \Lambda/|p|}\frac{\omega_k }{|k-\omega_p|}\,dk=\omega_p\int_{|k|\leq \Lambda/|p|}\frac{\omega_p\cdot\omega_k }{|k-\omega_p|}\,dk
\end{equation*}
which leads to $\cD^0(p)=v_{\rm eff}(p)\sig\cdot p$ with $v_{\rm eff}(p)$ as in~\eqref{eq:cD_0} and
\begin{align*}
g(R)&=\frac{1}{4\pi}\int_{|k|\leq R}\frac{\omega_p\cdot\omega_k }{|k-\omega_p|}\,dk\\
&=\frac{1}{4\pi}\int_0^{2\pi} \int_0^{R} \frac{\cos \theta}{\sqrt{r^2 - 2r \cos \theta + 1 }} r dr d\theta\\
&=\frac{1}{2\pi}\int_0^{\pi/2} \int_0^{R} \cos \theta\left(\frac{1}{\sqrt{r^2 - 2r \cos \theta + 1 }}-\frac{1}{\sqrt{r^2 + 2r \cos \theta + 1 }}\right) r dr d\theta.
\end{align*}
Since the integrand is non-negative, it is now clear that $g$ is increasing on $[1,\ii)$. For large $R$ we have
$$g(R)\underset{R\to\ii}{\sim}\frac{\log R}{\pi}\int_0^{\pi/2} \cos^2 \theta\, d\theta=\frac{\log R}{4}.$$

For the value of $g(1)$, we integrate first in $r$ and obtain
$$\int_0^1\frac{r}{\sqrt{r^2 - 2r \cos \theta + 1}} dr=-1 + \cos(\theta) \log\big(1 + \sin^{-1}(\theta/2)\big) + 2 \sin(\theta/2)$$
and we therefore get
\begin{align*}
g(1)&=\frac{1}{2\pi}\int_0^{\pi}\Big(\cos^2(\theta) \log\big(1 + \sin^{-1}(\theta/2)\big) + 2 \cos(\theta)\sin(\theta/2)\Big)\,d\theta\\
&=\frac{1}{2\pi}\int_0^{\pi}\cos^2(\theta) \log\big(1 + \sin(\theta/2)\big)\,d\theta-\frac{1}{2\pi}\int_0^{\pi}\cos^2(\theta) \log\big(\sin(\theta/2)\big)\,d\theta-\frac{2}{3\pi}.
\end{align*}
The result follows by explicit integration, using that
$${\rm G}=2\int_0^{\pi/4}\log(2\cos(\theta))\,d\theta=-\int_0^{\pi/4}\log(2\sin(\theta))\,d\theta.$$
This concludes the proof of the lemma.
\end{proof}

Because $\cD^0$ is equal to the original $D^0$ multiplied by $\vF+g(\Lambda/|p|)\geq g(1)>0$, we have
$$\gamma^0_{\rm ren}=-\frac{\cD^0}{2|\cD^0|}=\1_{(-\ii,0)}(\cD^0)-\frac{I}{2}.$$
In other words, the \emph{non-interacting} free Dirac sea solves the nonlinear equation of the \emph{interacting} system. 
This is in stark contrast with the results of~\cite{LieSie-00,HaiLewSol-07} in which the interacting Dirac sea was found to be very different from the non-interacting one, as we have already mentioned. With Lemma~\ref{lem:cD_0} at hand, we are now able to prove that $\gamma^0_{\rm ren}$ is indeed the global minimizer of $\cF$.

\begin{theorem}[Ground state of free graphene]\label{thm:free_vacuum}
Fix any ultraviolet cut-off $\Lambda>0$. If 
\begin{equation}
v_F\geq \frac 14 \frac{\Gamma(1/4)^2}{\Gamma(3/4)^2}-g(1)\simeq 2.0560,
\label{eq:cond_v_F}
\end{equation}
then the minimization problem
\[
\min\Big\{\cF(\gamma_{\rm ren})\ :\ \gamma_{\rm ren}=f_{\rm ren}(p),\ -\1_{B(0,\Lambda)}/2\leq f_{\rm ren}\leq \1_{B(0,\Lambda)}/2,\ \rho_{\gamma_{\rm ren}}\equiv0\Big\}
\]
with $\cF$ defined in~\eqref{eq:def_energy_per_unit_volume}, admits the unique minimizer 
$$\gamma^0_{\rm ren}=-\frac{D^0}{2|D^0|}=-\frac{\cD^0}{2|\cD^0|}.$$
\end{theorem}

We note that the case of graphene in which $\vF\simeq2.2$ is covered.

\begin{proof}
The proof is exactly the same as in the massive case~\cite{HaiLewSol-07}, since the argument of~\cite{HaiLewSol-07} relies on the fact that the Coulomb potential can 
be estimated by the effective mean-field operator, which does not require a positive mass. 

Going back to Formula~\eqref{eq:def_energy_per_unit_volume}, we write
\begin{multline*}
\cF(\gamma)-\cF(\gamma^0_{\rm ren})=\frac 1{(2\pi)^2} \bigg(\int_{B(0,\Lambda)} \Tr \big(\cD^0(p)\, (f(p)-f^0_{\rm ren}(p))\big)\\
 - \frac {1}2 \int_{\R^2} \frac{|(\check{f}-\check{f}^0_{\rm ren})(x)|^2}{|x|} dx  \bigg). 
\end{multline*}
Using that $-\1(\sigma\cdot p\leq0)\leq f(p)-f^0_{\rm ren}(p)\leq \1(\sigma\cdot p\geq0)$ exactly as in~\cite{HaiLewSol-07}, we see that 
\begin{align*}
\Tr \big(\cD^0(p)\, (f(p)-f^0_{\rm ren}(p))\big)&\geq \Tr \big(|\cD^0(p)|\, (f(p)-f^0_{\rm ren}(p))^2\big)\\
&\geq (\vF+g(1))\Tr \big(|p|\, (f(p)-f_{\rm ren}(p))^2\big)
\end{align*}
for a.e. $p\in B(0,\Lambda)$. Hence we conclude that
$$\cF(\gamma)-\cF(\gamma^0_{\rm ren})\geq \frac 1{(2\pi)^2}\left((\vF+g(1))\int_{\R^2}\big|(-\Delta)^{1/4}F|^2
 - \frac {1}2 \int_{\R^2} \frac{|F|^2}{|x|} dx  \right)$$
with $F:=\check{f}-\check{f}^0_{\rm ren}$. Kato's inequality in 2D tells us that
\begin{equation}
\frac{1}{|x|} \leq \frac 12 \frac{\Gamma(1/4)^2}{\Gamma(3/4)^2}\sqrt{-\Delta},
\label{eq:Kato}
\end{equation}
where the constant is optimal~\cite{Herbst-77,Yafaev-99} (see also~\cite[Lemma 8.2]{LieSei-09}). We deduce that $\cF(\gamma)\geq\cF(\gamma^0_{\rm ren})$, provided 
$\vF$ satisfies the inequality~\eqref{eq:cond_v_F}. Since there is no other optimizer than $0$ for Kato's inequality, even the equality in~\eqref{eq:cond_v_F} is covered.
\end{proof}

In summary, we have proved that $\gamma^0_{\rm ren}=P^0_--I/2$ is the unique minimizer of the energy per unit volume when $V\equiv0$. Arguing exactly as in~\cite{HaiLewSol-07}, it is then 
possible to prove that  $\gamma^0_{\rm ren}$ is also the thermodynamic limit of the true ground states of the Hartree-Fock energy, without the translation-invariance ansatz. The proof is even much easier than in~\cite{HaiLewSol-07} since $\gamma^0_{\rm ren}$ is known exactly and solves the self-consistent equation in a box as well. Instead of pursuing this route in detail, we accept that $\gamma^0_{\rm ren}$ is the actual free Dirac sea, and we now study local perturbations of it, in the spirit of~\cite{ChaIra-89,HaiLewSer-05a,HaiLewSer-05b}.

\section{Ground states of Hartree-Fock graphene in local external potentials}\label{sec:BDF}

Let us now come back to the Hartree-Fock energy~\eqref{HF_QED} and assume that the external field $V$ does not vanish. 
We will always make the assumption that $V$ is local in a sense to be made precise below, which puts us in a situation where we can think of the sought-after Hartree-Fock ground state as a local perturbation of the free Dirac sea. 

We consider any Hartree-Fock state described by its density matrix $\gamma$ (or equivalently by its renormalized density matrix $\gamma_{\rm ren}=\gamma-I/2$) and which we assume to be ``sufficiently close'' 
to $P^0_-$. The infinite volume Hartree-Fock energy of $\gamma_{\rm ren}$ is of course infinite, it is proportional to the volume exactly like the one of $\gamma_{\rm ren}^0=P^0_--I/2$. However we can, at least formally, subtract the (infinite) constant  $\cE_{\rm HF}^V(\gamma^0_{\rm ren})=\cE_{\rm HF}^0(\gamma^0_{\rm ren})$ and obtain a perfectly well-defined energy. 
A formal computation yields
\begin{equation*}
\cE_{\rm HF}^V(\gamma_{\rm ren})-\cE_{\rm HF}^V(\gamma^0_{\rm ren})=\cE^V_{\rm BDF}(Q),
\end{equation*}
where
$$Q=\gamma_{\rm ren}-\gamma^0_{\rm ren}=\gamma-P^0_-$$
and where $\cE^V_{\rm BDF}$ is the so-called \emph{Bogoliubov-Dirac-Fock} energy, formally defined by
\begin{multline}
\cE^V_{\rm BDF}(Q)=\tr(\cD^0 Q)+\int_{\R^2}V(x)\rho_Q(x)\,dx+\frac12  \iint_{\R^2\times\R^2}\frac{\rho_Q(x)\,\rho_Q(y)}{|x-y|}dx\,dy\\
-\frac12  \iint_{\R^2\times\R^2}\frac{|Q(x,y)|^2}{|x-y|}dx\,dy.
\label{eq:formal_BDF}
\end{multline}
Again the energy functional looks like the usual Hartree-Fock energy, with the difference that $\cD^0$ now appears instead of $D^0$ 
and that it is applied to the operator $Q$ which is a \emph{difference} of two Hartree-Fock density matrices. The operator $Q$ satisfies the constraint
$$-P^0_-\leq Q\leq 1-P^0_-=P^0_+.$$

\begin{remark}To our knowledge, the idea of subtracting the infinite energy of the free Dirac sea in order to get a bounded below energy, was used for the first time in~\cite{HaiLewSol-07}. This was generalized to positive temperature in~\cite{HaiLewSei-08}. In previous works~\cite{ChaIra-89,HaiLewSer-05a} dealing with the Hartree-Fock approximation of QED, another justification based on normal ordering was employed. 
\end{remark}
\begin{remark}
Let us mention that the perturbed state $\gamma$ can always be seen as a Bogoliubov rotation of the free Dirac sea in its Fock representation, which is why Chaix and Iracane used the name `Bogoliubov' for the energy~\eqref{eq:formal_BDF}. We could as well call it a \emph{relative Hartree-Fock energy} but we prefer to keep the name Bogoliubov-Dirac-Fock (BDF) for historical reasons. 
\end{remark}

Our tasks in this section are then to prove that:
\begin{enumerate}
\item[(a)] $Q=0$ is the unique minimizer of $\cE_{\rm BDF}^0$ for $V\equiv0$, which is a ``local'' version of the fact that the free Dirac sea $P^0_-$ is the unique ground state of the system without external field;
\item[(b)] if $V\neq0$, then there exists a ground state for $\cE_{\rm BDF}^V$ which solves the self-consistent Hartree-Fock equation.
\end{enumerate}
Before turning to these problems, we however have to properly define the BDF energy. It is now well understood that the Hartree-Fock ground state $\gamma$ of an infinite Coulomb system can in general behave badly. Usually $Q=\gamma-P^0_-$ is not a trace-class operator and its density is sometimes not in $L^1$, see~\cite{GraLewSer-09,CanLew-10}. It has long-range oscillations which are not integrable at infinity in some cases. For these reasons, it is not fully obvious to give a clear meaning to the BDF energy~\eqref{eq:formal_BDF} and to find a suitable class of states in which minimizers will be. Following ideas from~\cite{HaiLewSer-05a}, we introduce the correct functional analysis setting in the next section.

\subsection{Function spaces and definition of the density}

Given an operator $Q$, we define $Q^{\epsilon\epsilon'}:=P^0_\epsilon QP^0_{\epsilon'}$ where $\epsilon,\epsilon'\in\{\pm\}$. Our starting point is the remark that, for a nice-enough operator $Q$ (say finite rank),
$$\tr(\cD^0 Q)=\tr\big(|\cD^0| (Q^{++}-Q^{--}\big)\geq \tr|\cD^0|Q^2.$$
Here we have used that $P^0_-$ commutes with $\cD^0$ and that 
$$-P^0_-\leq Q\leq P^0_+\quad\Longleftrightarrow\quad Q^2\leq Q^{++}-Q^{--},$$
as was remarked first in~\cite{BacBarHelSie-99}. We see that a state will have a finite relative kinetic energy when $|\cD^0|^{1/2} Q^{\pm\pm}|\cD^0|^{1/2}$ are trace-class, but we cannot gain any other information on $Q^{\pm\mp}$ than $|\cD^0|^{1/2}Q^{\pm\mp}$ being Hilbert-Schmidt. Thus we shall assume that 
\begin{equation}
|\cD^0|^{1/2} Q^{\pm\pm}|\cD^0|^{1/2}\in\gS^1\quad\text{and}\quad Q^{\pm\mp}|\cD^0|^{1/2}\in\gS^2
\label{eq:assumptions_Q} 
\end{equation}
where $\gS^p$ denotes the usual $p$-th Schatten space ($\gS^1$ and $\gS^2$ are respectively the spaces of trace-class and Hilbert-Schmidt operators). This enables us to properly define
$$\tr(\cD^0 Q):=\tr\big(|\cD^0|^{1/2} (Q^{++}-Q^{--})|\cD^0|^{1/2}\big).$$
Our next task is to give a clear definition of the density $\rho_Q$ under the very weak assumptions~\eqref{eq:assumptions_Q} on $Q$. To simplify our exposition we introduce the Banach space 
\begin{equation*}
\X := \Big\{ Q\in\cB(\gH_\Lambda)\ :\ Q^* = Q,\ |p|^{1/2}Q ^{\pm \pm }|p|^{1/2} \in \gS^1,\ Q|p| \in \gS^2\Big\},
\end{equation*}
where we recall from \eqref{eq:HScut} that $\gH_\Lambda$ is the Hilbert space of $L^2$ functions with compact support in $B(0,\Lambda)$ in the Fourier domain. The set of bounded operators on this space is then 
denoted by $\cB(\gH_\Lambda)$. We also introduce
\begin{equation*}
\tilde\X := \Big\{ Q\in\cB(\gH_\Lambda)\ :\ Q^* = Q,\ |\cD^0|^{1/2}Q ^{\pm \pm }|\cD^0|^{1/2} \in \gS^1,\ Q|\cD^0| \in \gS^2\Big\}\subset \X, 
\end{equation*}
since $|\cD^0|\geq (\vF+g(1))|D^0|=(\vF+g(1))|p|$, as we have shown before. For what follows it will be 
convenient to work with states in $\X$. By doing so, we ignore the logarithmic divergence at 0 of $\cD^0(p)$, which is an additional information for us to be used in due time.
Let us remark that $\tilde\X$ contains elements which are not even compact. This makes the mathematics more involved than in the situation where the particles have nonzero mass. 

Next, we shall show that any $Q\in\X$ indeed has a well defined density $\rho_Q$. We start by remarking that $\rho_Q$ is locally well defined.
\begin{lemma}[Definition of the density $\rho_Q$]\label{lem:L1_loc}
Any $Q \in \X$ is locally trace class, i.e., for all $\chi(x) \in L^\ii(\R^2)$ with compact support, $\chi Q \chi \in \gS^1$. Moreover, the density $\rho_Q$ is in $L^\ii(\R^2)$.  
\end{lemma}
\begin{proof}
Due to the cut-off we have
$$ \chi Q \chi = \chi \Pi_{\Lambda} Q \Pi_{\Lambda} \chi$$
where $\Pi_\Lambda=\1(|p|\leq\Lambda)$.  
Since $\chi \in L^2$, then clearly $\chi \Pi_{\Lambda} \in \gS^2,$ such that together with the boundedness of $Q$ we 
get $\chi Q\chi \in \gS^1$, as stated. This proves that $\rho_Q$ is a well defined function in $L^1_{\rm loc}$. To see that it is actually uniformly bounded, we use that $Q$ is self-adjoint and the ultraviolet cut-off $\Lambda$, to infer
$ - \|Q\| \Pi_\Lambda \leq Q \leq \|Q\| \Pi_\Lambda$.
On the diagonal we get 
$|\rho_Q(x)| \leq  \|Q\| \rho_{\Pi_\Lambda}(x)=\|Q\| {\Lambda^2}/{(4\pi)}$.
\end{proof}

\begin{remark}\label{rmk:weak_CV}
Any bounded sequence $(Q_n)$ in $\X$ has a weakly$-\ast$ convergent subsequence, $Q_{n_k}\wto Q$ in the sense that 
$$\tr(AQ_{n_k})\to \tr(AQ),\quad\text{$\forall A\in\gS^1$},$$
$$\tr(K|p|^{1/2}Q_{n_k}^{\pm\pm}|p|^{1/2})\to \tr(K|p|^{1/2}Q^{\pm\pm}|p|^{1/2}),\quad \text{$\forall K$ compact},$$
and 
$$\tr(BQ_{n_k}|p|^{1/2})\to \tr(BQ|p|^{1/2}),\quad \text{$\forall B\in\gS^2$}.$$ 
Using that $\chi\Pi_\Lambda\in\gS^2$ when $\chi\in L^\ii_c$, it is then elementary to verify that $\chi Q_{n_k}\chi\to \chi Q\chi$ strongly in the trace-class. This implies that $\rho_{Q_{n_k}}\to\rho_Q$ strongly in $L^{1}_{\rm loc}$, hence strongly in $L^p_{\rm loc}(\R^2)$ for all $1\leq p<\ii$, by interpolation.
\end{remark}

Next, we introduce the so-called {\it Coulomb space}:
$$\cC := \left\{ \phi\ :\  D(\phi,\phi):=2\pi\int_{\R^2} \frac{|\hat \phi(k)|^2}{|k|}\,dk < \infty\right\}$$ 
which is the natural energy space for the density $\rho_Q$. If we decompose the density of $Q$ into 
$$\rho_Q = \rho_{Q^{++}} +  \rho_{Q^{--}} + \rho_{Q^{+-}} + \rho_{Q^{-+}},$$
then we can deduce the following properties for the elements of $\X$. 

\begin{lemma}[The density is in $\cC$]\label{lem:Coulomb}
Assume that $Q\in\X$ is such that $-P^0_-\leq Q\leq P^0_+$. Then $\rho_{Q^{\pm \pm}} \in L^p(\R^2)$ for $3/2 \leq p \leq \infty$ and $\rho_{Q^{\pm\mp}}\in\cC $, which particularly implies that $\rho_Q \in \cC$. 
\end{lemma}

\begin{proof}
Since $0\leq \pm Q^{\pm\pm}\leq 1$, the Lieb-Thirring inequality in 2D for the relativistic kinetic energy~\cite{LieSei-09} immediately yields
$$ \pm\tr\big(|p| Q^{\pm\pm}\big) \geq c \int_{\R^2} |\rho_{Q^{\pm\pm}}(x)|^{3/2} dx.$$
Since we already know that $\rho_{Q^{\pm\pm}}\in L^\ii(\R^2)$, $\rho_{Q^{\pm\pm}}\in\cC$ now follows from the Hardy-Littlewood-Sobolev inequality in 2D, see, e.g.,~\cite{LieLos-01}.
The second part of the statement follows by a duality argument, provided that we can show  
$$
|\tr \xi Q^{+-}| = | \langle \xi, \rho_{Q^{+-}}\rangle| \leq c  \left(\int_{\R^2}|k| |\hat \xi(k)|^2 dk\right)^{1/2},
$$
for all  $\xi\in C^\ii_c(\R^2)$.
To this end we first estimate
\begin{align*}
\tr \xi Q^{+-} &= \tr\left( \frac {P_+^0}{|p|^{1/4}} \xi \frac{P_-^0}{|p|^{1/4}} |p|^{1/4} Q |p|^{1/4} \right) \leq \left \| \frac {P_+^0}{|p|^{1/4}} \xi \frac{P_-^0}{|p|^{1/4}} \right\|_{\gS^2} \||p| Q \|_{\gS^2}.
\end{align*}
Using
$$ \Tr P_+^0(p)P_-^0(q) = \frac{\omega_p \cdot \omega_q}{|p||q|} - 1,$$
we can compute 
\begin{align*}
&\left \| \frac {P_+^0}{|p|^{1/4}} \xi \frac{P_-^0}{|p|^{1/4}} \right\|_{\gS^2}^2=\\
& \qquad= \frac 1{(2\pi)^2} \iint_{|p|,|q| \leq \Lambda} \frac{|\xi(p-q)|^2 \Tr P_+^0(p)P_-^0(q)}{|p|^{1/2} |q|^{1/2}} \, dp\, dq \\ 
&\qquad\leq \frac 1{(2\pi)^2} \iint_{\R^2\times \R^2}  |\xi(k)|^2 \frac{|\ell+k/2||\ell-k/2| - (\ell+k/2)\cdot (\ell-k/2)}{|\ell+k/2|^{3/2}|\ell-k/2|^{3/2} } \, dk \, d\ell \\
&\qquad= \int_{\R^2} dk |\hat \xi(k)|^2 |k| \int_{\R^2} \frac{|\ell+\omega_k/2||\ell-\omega_k/2| - (\ell+\omega_k/2)\cdot (\ell-\omega_k/2)}{|\ell+\omega_k/2|^{3/2}|\ell-\omega_k/2|^{3/2} }\, d\ell  \\ 
&\qquad\leq \text{c} \int |\hat \xi(k)|^2 |k|\, dk .
\end{align*}
Here $\omega_k:=k/|k|$ and we have used the fact that
$$\frac 1{(2\pi)^2}  \int  \frac{|\ell+\omega_k/2||\ell-\omega_k/2| - (\ell+\omega_k/2)\cdot (\ell-\omega_k/2)}{|\ell+\omega_k/2|^{3/2}|\ell-\omega_k/2|^{3/2} }\, d\ell   = \text{c}, $$ 
is a finite integral, independent of the direction $\omega_k\in S^1$.
\end{proof}

\subsection{Stability of the free Dirac sea}

We have shown that the density $\rho_Q$ is in the Coulomb space $\mathcal C$ whenever 
$Q\in\X$ and $-P^0_-\leq Q\leq P^0_+$. Thus we see that the direct term is well defined for all such $Q$ with finite relative kinetic energy. We can now define the (free) BDF energy as
$${\cE^0_{\rm BDF}(Q):=\tr(\cD^0Q)+\frac{1}{2}D(\rho_Q,\rho_Q)-\frac{1}{2} \iint_{\R^2\times\R^2}\frac{|Q(x,y)|^2}{|x-y|}dx\,dy.}$$
where we recall that
$$D(\rho,\rho'):= \iint_{\R^2\times\R^2}\frac{\rho(x)\rho'(y)}{|x-y|}dx\,dy= 2\pi \int_{\R^2}\frac{\overline{\widehat{\rho}(k)}\widehat{\rho}'(k)}{|k|}dk$$
is the so-called Coulomb scalar product.

The following lemma shows that the exchange term is also well defined and that the BDF energy $\cE^0_{\rm BDF}$ is non-negative with $Q=0$ being its unique minimizer. 
Recalling that $\cE^0_{\rm BDF}$ is the relative energy counted with respect to the free state $P^0_-$ and that $Q=\gamma-P^0_-$, this consequently 
shows that the free ground state $P^0_-$ is stable under local deformations. Here `local' refers to perturbations such that $Q\in\X$ but not necessarily small in norm.

\begin{lemma}[Stability of $P^0_-$]
For any fixed ultraviolet cut-off $\Lambda$, the mapping $Q\mapsto \cE^0_{\rm BDF}(Q)$ is well defined and continuous on $\tilde\X$.
If $\vF$ satisfies 
\begin{equation}
v_F\geq \frac 14 \frac{\Gamma(1/4)^2}{\Gamma(3/4)^2}-g(1),
\label{eq:cond_v_F_bis}
\end{equation}
then we have
$$0\leq \cE^0_{\rm BDF}(Q)<\ii,$$
for all $-P^0_-\leq Q\leq P^0_+$ with $|\cD^0|^{1/2}Q^{\pm\pm}|\cD^0|^{1/2}\in\gS^1$. Furthermore, $\cE^0_{\rm BDF}(Q)=0$ if and only if $Q\equiv0$.
\end{lemma}
\begin{proof}
Using Kato's inequality~\eqref{eq:Kato} we infer, following \cite{BacBarHelSie-99} and similarly as in the proof of Theorem~\ref{thm:free_vacuum},
$$
\frac{1}{2}\iint\frac{|Q(x,y)|^2}{|x-y|}dx\,dy\leq \frac1 4  \frac{\Gamma(1/4)^2}{\Gamma(3/4)^2} \tr(|p|Q^2)\leq   \frac{\Gamma(1/4)^2}{4\Gamma(3/4)^2(\vF+g(1))}\tr(|\cD^0|Q^2),
$$
which is well-defined due to our assumption that $Q|\cD^0|^{1/2}\in\gS^2$. In addition, since $-P^0_-\leq Q\leq P^0_+$, we have $\tr(|\cD^0|Q^2)\leq \tr(\cD^0Q)$ and therefore 
\begin{equation}
\frac{1}{2}\iint\frac{|Q(x,y)|^2}{|x-y|}dx\,dy\leq \frac{\Gamma(1/4)^2}{4\Gamma(3/4)^2(\vF+g(1))} \tr(\cD^0 Q).  
\label{estim_exchange}
\end{equation}
In other words, if $|\cD^0|^{1/2} Q^{\pm\pm}|\cD^0|^{1/2}\in \gS^1$ and if the Fermi velocity $\vF$ satisfies~\eqref{eq:cond_v_F_bis}, the exchange term is controlled by the kinetic energy. We have already seen that $\rho_Q\in\cC$. Using then that $D(\rho_Q,\rho_Q)\geq0$, we see that $\cE^0_{\rm BDF}(Q)\geq0$ and that $Q=0$ is the unique minimizer of $\cE^0_{\rm BDF}$, as stated.
\end{proof}

\subsection{Existence of minimizers for graphene in an external field} 

In the next step, we shall submit our graphene sheet to an external electrostatic field of the form
$$V=-\nu\ast\frac{1}{|x|}$$
with $\nu$ the density of charge of the defect. The corresponding BDF energy now reads
$$\cE_{\rm BDF}^V(Q):=\cE_{\rm BDF}^0(Q)-D(\rho_Q,\nu).$$
Using our estimate~\eqref{estim_exchange} on the exchange term and that $D(\cdot,\cdot)$ defines a scalar product, we immediately get the lower bound
$$\cE_{\rm BDF}^V(Q)\geq -\frac12 D(\nu,\nu),$$
provided that $\vF$ satisfies~\eqref{eq:cond_v_F_bis}.
Therefore the BDF energy is bounded from below when $\nu\in\cC$ which is our way of measuring the locality of the potential $V$.
This enables us to consider the minimization problem
\begin{multline}
E_{\rm BDF}^V:=\inf\Big\{\cE^V_{\rm BDF}(Q)\ :\\
 Q\in \cB(\gH_\Lambda),\ -P^0_-\leq Q=Q^*\leq P^0_+,\ |\cD^0|^{1/2}Q^{\pm\pm}|\cD^0|^{1/2}\in\gS^1\Big\}.
\label{eq:def_min_BDF}
\end{multline}
We shall prove the existence of a corresponding minimizer, which is the Hartree-Fock ground state of interacting graphene in the presence of defects. 
This existence result is non-trivial for the simple reason that we have no mass. 

\begin{theorem}[Existence of a ground state for infinite volume graphene with defects] \label{thm:existence_pol_vac}
Fix $\Lambda >0$ and let $\vF$ be such that
\begin{equation}
v_F >  \frac 14 \frac{\Gamma(1/4)^2}{\Gamma(3/4)^2}-g(1).
\label{eq:cond_v_F_bis_strict}
\end{equation}
For any $\nu$ in the Coulomb space $\cC$, the problem~\eqref{eq:def_min_BDF} admits at least one minimizer $Q$, satisfying the self-consistent equation 
\begin{equation}
\label{polvac}
Q + P^0_- = \1_{\mathcal I}\left(\cD^0 - \nu \ast \frac 1{|x|} + \rho_{Q}\ast \frac 1{|x|} - \frac{Q(x,y)}{|x-y|}\right)
\end{equation}
where $\mathcal I=(-\infty, 0 )$ or $\mathcal I=(-\infty, 0]$. Equivalently, with $\gamma=Q+P^0_-$ denoting the density matrix of the optimal HF state, we have
\[
\gamma= \1_{\mathcal I}\left(\vF\,\sig\cdot (-i\nabla) - \nu \ast \frac 1{|x|} + \rho_{\gamma-I/2}\ast \frac 1{|x|} - \frac{(\gamma-I/2)(x,y)}{|x-y|}\right). 
\]
\end{theorem}

Remark that HF ground states of graphene in the presence of a defect can be chosen pure and with the Fermi level either filled or unfilled completely, as is usual for Hartree-Fock theories. This follows from Lieb's variational principle~\cite{Lieb-81} and the no-unfilled shell theorem of Bach, Lieb, Loss and Solovej~\cite{BacLieLosSol-94}. Note the strict inequality in~\eqref{eq:cond_v_F_bis_strict}. There are probably also ground states in the case of equality, but we shall 
not consider this case for the sake of simplicity.

The proof of theorem~\ref{thm:existence_pol_vac} is a bit long and it will be given in Section \ref{sec:proof} below. 
The method is similar to the one in \cite{HaiLewSer-05b}, but several modifications are needed due to the absence of the mass. 
One important additional input is a localization estimate inspired by \cite{LenLew-10} and which is detailed in the Appendix. 
In contrast to~\cite{HaiLewSer-05b} we shall use that the kinetic energy has an {\em infinite} velocity for $p =0$, which induces a better control of the exchange term. Our method does not seem to apply otherwise.

\begin{remark} We do not know if the optimal state $Q$ has a finite (relative) number of particles, that is, $Q^{\pm\pm}$ might be not trace-class and the charge $\tr(Q):=\tr(Q^{++}+Q^{--})$ could be infinite. The operator $Q$ could even be non compact in general, because of the absence of a gap. This would mean that the corresponding minimizers $Q$ live in a Fock representation which is not equivalent to that of $P^0_-$, even if the relative energy is itself finite, by the Shale-Stinespring theorem~\cite{Thaller}. In~\cite{HaiLewSer-09} another minimization problem consisting in fixing the relative charge $\tr(Q)$ was considered. Because there is no gap and $\tr(Q)$ can be infinite, an analogous approach does not seem to make sense in our context. 
Bound states with finitely many electrons have been constructed in a projected Dirac-Fock-type model in~\cite{EggMarSieSto-10} but there a magnetic field is used to confine the particles and create a gap.
\end{remark}

\section{Linear response to an applied external field}\label{sec:linear_response}

In this section we consider a small external field 
$$V_\lambda=-\lambda\nu\ast\frac{1}{|x|},\qquad \lambda\ll1$$
and discuss the linear response of graphene within our Hartree-Fock theory. To this end, we denote by $Q_\lambda$ a chosen minimizer for each $\lambda$ and note that 
$\cE^{V_\lambda}_{\rm BDF}(Q_\lambda)\leq0$,
which is seen by using $Q\equiv0$ as a trial state. From this we deduce that
$$\tr\cD^0 Q_\lambda -\frac12 \iint_{\R^2\times\R^2}\frac{|Q_\lambda(x,y)|^2}{|x-y|}dx\,dy+\frac{1}{2}\norm{\rho_\lambda-\lambda\nu}_\cC^2\leq \frac{\lambda^2}{2}\norm{\nu}_\cC^2.$$
Assuming the strict inequality
\begin{equation}
v_F >  \frac 14 \frac{\Gamma(1/4)^2}{\Gamma(3/4)^2}-g(1),
\label{eq:cond_v_F_bis_strict2}
\end{equation}
which allows to control the exchange term, in view of \eqref{estim_exchange}, we deduce that
$$\tr|\cD^0| Q_\lambda^2+ \iint_{\R^2\times\R^2}\frac{|Q_\lambda(x,y)|^2}{|x-y|}dx\,dy+\norm{\rho_\lambda-\lambda\nu}_\cC^2=O(\lambda^2),$$
with $\rho_\lambda:=\rho_{Q_\lambda}$. This confirms that $Q_\lambda$ is of order $\lambda$ in $\tilde\X$.

Recall from Theorem~\ref{thm:existence_pol_vac} that $Q_\lambda$ satisfies the self-consistent equation
$$Q_\lambda=\1_{(-\ii,0)}(\cD_{Q_\lambda})-P^0_-=\1_{(-\ii,0)}(\cD_{Q_\lambda}) - \1_{(-\ii,0)}(\cD^0).$$
For simplicity we assume here that $\ker(\cD_{Q_\lambda})=\{0\}$ for all $\lambda$.
Denoting $$\cD_{Q_\lambda} = \cD^0 - \lambda\nu \ast \frac 1{|x|} + \rho_{\lambda} \ast \frac 1{|x|} - \frac{Q_\lambda(x,y)}{|x-y|} = \cD^0 - \varphi_\lambda - R_\lambda,$$
with $ \varphi_\lambda = \lambda\nu \ast \frac 1{|x|} -  \rho_{\lambda} \ast \frac 1{|x|} ,$ and $R_\lambda= Q_\lambda(x,y)/|x-y|$,
we can write, using a formula of Kato, that
$$ Q_\lambda= \frac 1{2\pi} \int_{-\infty}^\infty \left[ \frac{1}{\cD_{Q_\lambda} + i \eta} - \frac{1}{\cD^0+ i \eta} \right] d\eta.$$
In view of the resolvent formula we deduce that 
$$ Q_\lambda\equiv Q_{1, D} + Q_{1,X}= \frac 1{2\pi} \int_{-\infty}^\infty \left[ \frac{1}{\cD^0 + i \eta}(\phi_\lambda+R_\lambda)\frac{1}{\cD^0+ i \eta} \right] d\eta +O(\lambda^2).$$
The remainder $O(\lambda^2)$ has to be estimated carefully, but for convenience we remain formal in this discussion. Because $\phi_\lambda$ and $R_\lambda$ are themselves affine in $Q_\lambda$, the first order term for $Q_\lambda$ is obtained by inverting these linear maps. There is no simple expression for it. However, following~\cite{HaiSie-03,HaiLewSer-05a}, it is possible to compute explicitly the density coming from the term involving $\phi_\lambda$, as will be explained now:

We look at the density associated with the operator stemming from the direct term
$$Q_{1,D} =  \frac 1{2\pi} \int_{\infty}^\infty \frac{1}{\cD^0+ i \eta} \varphi_\lambda \frac{1}{\cD^0+ i \eta} d\eta .$$
In Fourier variables this reads
$$ \hat Q_{1,D} (p,q) = \frac 1{(2\pi)^2} \int_{\infty}^\infty \frac{1}{\cD^0(p)+ i \eta} \hat \varphi_\lambda(p-q) \frac{1}{\cD^0(q)+ i \eta} d\eta,$$
Using 
$$M(p,q) = \frac 1\pi \int_{\infty}^\infty \frac{1}{\cD^0(p)+ i \eta} \frac{1}{\cD^0(q)+ i \eta} d\eta = \frac 1{E(p) + E(q)}(\sig \cdot \omega_p \sig \cdot \omega_q - 1),$$
with $$ E(p) = |p| \left(\vF + g\left(\frac{\Lambda}{|p|}\right)\right),$$
we can write
$$\hat Q_{1,D} (p,q) = \frac 1{4\pi} \hat \varphi (p-q) M(p,q).$$
Using $ \cD^0(p)/|E(p)| = \sig \cdot \omega_p$, we see that the corresponding density reads 
\begin{align*}
\hat \rho_{1,D} (k) &= \frac 1{2\pi} \int_{|p+k/2| \leq \Lambda, |p-k/2| \leq \Lambda } \Tr \hat Q_{1,D}(p+k/2,p-k/2) dp \\ 
&= \frac 1{2\pi} \hat \varphi_\lambda(k) |k| B(k) = \big(\hat\nu(k)-\hat{\rho}_\lambda(k)\big)  B(k),
\end{align*}
where we have used that in 2D
$$\widehat{\rho \ast \frac 1{|x|}}(k) = 2\pi \hat \rho(k) \frac 1{|k|}$$ 
and we denote
\begin{align*}\label{Blam} 
B(k)& = - \frac 1{|k|2\pi} \int_{\substack{|p+k/2| \leq \Lambda\\ |p-k/2| \leq \Lambda}}  \frac{ (p+ k/2)\cdot (p-k/2) - |p+k/2||p-k/2|}{|p+k/2|\;|p-k/2|((E(p+k/2) + E(p-k/2))} \, dp \nonumber\\ 
& =  \frac {1}{2\pi} \int_{\substack{|p+\omega_k/2| \leq \Lambda/|k|\\ |p-\omega_k/2| \leq \Lambda/|k|}} \frac{ - (p+ \omega_k/2)\cdot (p-\omega_k/2) + |p+\omega_k/2||p-\omega_k/2|}{|p+\omega_k/2|\;|p-\omega_k/2|}\times\nonumber\\
&\quad\times \frac{1}{|p+\omega_k/2|\left(\vF+g\left(\frac{\Lambda}{|k|\,|p+\omega_k/2|}\right)\right)+|p-\omega_k/2|\left(\vF+g\left(\frac{\Lambda}{|k|\,|p-\omega_k/2|}\right)\right)}\,dp.
\end{align*}

\begin{remark} 
In the case of three spatial dimensions there is a similar function $B(k)$ playing an important role. In this case, the 
value $B(0)>0$ is logarithmically divergent with respect to $\Lambda$, which is the reason for the requirement of charge renormalization, giving rise to the Uehling potential, see~\cite{HaiSie-03,HaiLewSer-05b,GraLewSer-09,GraLewSer-11}. 
\end{remark}

In Fourier space, we can write the self-consistent equation as in~\cite{HaiLewSer-05b}
$$\widehat{\rho_\lambda}(k)=B(k)\big(\widehat{\nu}(k)-\widehat{\rho_\lambda}(k)\big)+\widehat{\rho}_{1,X}(k)+O(\lambda^2),$$
where $\rho_{1,X}=O(\lambda)$ is the first-order density coming from the exchange term 
$$Q_{1,X} =  \frac 1{2\pi} \int_{\infty}^\infty \frac{1}{\cD^0+ i \eta} R_\lambda \frac{1}{\cD^0+ i \eta} d\eta .$$
We see that
$$\widehat{\rho_\lambda}(k)=\lambda\frac{B(k)}{1+B(k)}\widehat{\nu}(k)+\frac{\widehat{\rho}_{1,X}(k)}{1+B(k)}+O(\lambda^2).$$
It is therefore natural to ask about the behavior of $B(k)$ for low momenta. If we neglect the density $\rho_{1,X}$ stemming 
from the exchange term, this will determine the decay in $x$ space of the Coulomb potential of the (first order) polarized graphene in presence of the external density $\nu$.

In order to answer this question, we simplify the expression of the function $B(k)$. First we remark that $B(k)$ is obviously radial, hence we can take $\omega_k=e_1:=e$, such that 
\begin{multline}
B(k) =  \frac {1}{2\pi} \int_{\substack{|p+e/2| \leq \Lambda/|k|\\ |p-e/2| \leq \Lambda/|k|}}  \frac{|p+e/2|\,|p-e/2| - p^2 + 1/4 }{|p+e/2||p-e/2|}\times\\
\times  \frac{1}{|p+e/2|\left\{\vF+g\left(\frac{\Lambda}{|k|\,|p+e/2|}\right)\right\}+|p-e/2|\left\{\vF+g\left(\frac{\Lambda}{|k|\,|p-e/2|}\right)\right\}}dp.
\end{multline} 
As in \cite{PauRos-36,GraLewSer-09}, we use the following change of variables,
\begin{equation*}
v = \frac{|p+e/2| - |p-e/2|}2,\qquad
w= \frac{|p+e/2| + |p-e/2|}2.
\end{equation*}
Denoting $ p = (x,y)$, this reads 
\begin{align*}
v = & \frac{\sqrt{(x + 1/2)^2 + y^2} - \sqrt{(x-1/2)^2 + y^2}}2,\\
w= & \frac{\sqrt{(x + 1/2)^2 + y^2} + \sqrt{(x-1/2)^2 + y^2}}2.
\end{align*}
The corresponding Jacobian is 
$$ \left| \frac{\partial(v,w)}{\partial(x,y)} \right|  = \frac{y}{2\sqrt{(x + 1/2)^2 + y^2}  \sqrt{(x-1/2)^2 + y^2}} = \frac{|p_2|}{2|p+e/2| \; |p-e/2| } .$$
We collect the following relations 
$$2w^2 + 2v^2 = (w + v)^2 + (w-v)^2 = 2(p^2 +1/4), $$
$$4 vw = (w+v)^2 - (w-v)^2 = 2 p \cdot e = 2 p_1 = 2 x, $$
and
\begin{multline*}
 |y| = |p_2| = \sqrt{ p^2 - (p\cdot e)^2} = \sqrt{ w^2 - 1/4 - 4 v^2(w^2 -1/4)}\\ = 2 \sqrt{ w^2 - 1/4} \sqrt{1/4 - v^2}. 
\end{multline*}
Observe that 
$$ w \geq  \frac{|(p+e/2) + ( p-e/2)|}2 = 1/2,$$
and $$ |v| \leq \frac{|(p+e/2) - (p-e/2)|}2 = 1/2.$$ 
Our constraints on $p$, i.e., $|p+\omega_k/2| \leq \Lambda/|k|$ and $|p-\omega_k/2| \leq \Lambda/|k|$ can be expressed in terms of $w,v$ as follows
$$ w \leq \Lambda/|k| \quad\text{and}\quad |v| \leq w - \Lambda/|k|,$$
such that 
\begin{equation*}
\left\{\begin{array}{l}
1/2 \leq w \leq \Lambda/|k|,\\
|v| \leq \min\{1/2, \Lambda/|k| - w \}.
\end{array}\right.
\end{equation*}
Changing variables we find 
\begin{align*}
&\frac{  |p+e/2||p-e/2| - p^2 + 1/4 }{|p+e/2||p-e/2|} dx\,dy \\
&\qquad\qquad=\frac{ (v+ w) (w-v) - (v^2 + w^2 -1/4)+1/4 }{ |p+e/2||p-e/2|}\ \frac{2  |p+e/2||p-e/2| }{|y| } dv\, dw\\
&\qquad\qquad= \frac{ 2 (1/4-v^2) }{\sqrt{ w^2 - 1/4} \sqrt{1/4 - v^2}}  dv\,dw=\frac{2\sqrt{1/4-v^2} }{\sqrt{ w^2 - 1/4} }  dv\,dw.
\end{align*}
On the other hand we can express
\begin{multline*}
|p+\omega_k/2|\left\{\vF+g\left(\frac{\Lambda}{|k|\,|p+\omega_k/2|}\right)\right\}+|p-\omega_k/2|\left\{\vF+g\left(\frac{\Lambda}{|k|\,|p-\omega_k/2|}\right)\right\}\\
=(v+w)\left\{\vF+g\left(\frac{\Lambda}{|k|\,(v+w)}\right)\right\}+(w-v)\left\{\vF+g\left(\frac{\Lambda}{|k|\,(w-v)}\right)\right\}
\end{multline*}
to arrive at
\begin{multline}
B(k) =  \frac 2\pi 
\int_{1/2}^{\Lambda/|k|}dw \int_{0}^{\min\{ 1/2, \Lambda/|k| - w\} }dv\, \frac{  \sqrt{1/4-v^2} }{\sqrt{ w^2 - 1/4} }\times \\ \times \frac{ 1 }{ 2\vF w + (v+w)g\left(\frac{\Lambda}{|k|\,(v+w)}\right)+(w-v)g\left(\frac{\Lambda}{|k|\,(w-v)}\right)}.
\label{eq:last_formula_B}
\end{multline}

\begin{remark}[The no-exchange case]\label{rmk:B_no_exchange}
If we discard the exchange term, then we have exactly the same calculation with $g$ replaced by 0 everywhere. In this case the linear response involves the modified function 
\begin{equation}
B^0(k)  = \frac 1{\pi\vF} 
\int_{1/2}^{\Lambda/|k|}dw \int_{0}^{\min\{ 1/2, \Lambda/|k| - w\} } \frac{  \sqrt{1/4-v^2} }{ w \sqrt{ w^2 - 1/4} } dv .
\end{equation}
Let now $ w = t/2$, and $v = (\cos \theta) /2$, such that $ 0 \leq \cos \theta \leq \min\{ 1 , 2\Lambda /|k| - t \}$. Then 
$$ B^0(k) = \frac 1{2\pi\vF} \int_1^{2 \Lambda /|k|} dt \int_0^{\arccos[\min\{ 1, 2\Lambda/|k| - t \}] } \frac 1{t \sqrt{t^2 -1} }\sin^2 \theta \, d\theta .$$
Since 
$2\Lambda/|k| - t < 1$ is equivalent to $t > 2\Lambda/|k| - 1$ we can decompose it into two integrals
\begin{align*}
B^0(k) &= \frac 1{2\pi\vF} \Bigg(\int_1^{2 \Lambda /|k|-1} dt \int_0^{\pi /2} \frac 1{t \sqrt{t^2 -1} }\sin^2 \theta d\theta \\ 
&\qquad \qquad\qquad\qquad + \int_{2\Lambda/|k| -1}^{2\Lambda/|k|}  dt     \int_{\arccos[ 2\Lambda/|k| - t]}^{\pi/2} \frac 1{t \sqrt{t^2 -1} }\sin^2 \theta d\theta \Bigg)\\ 
&= \frac 1{2\pi\vF} \frac{\pi} 4 \arccos(1/(2\Lambda/|k| -1))+ \frac 1{2\pi\vF} \int_{2\Lambda/|k| -1}^{2\Lambda/|k|}  dt \frac 1{4t \sqrt{t^2 -1} }\times\\
&\qquad\times \left(-2 \arccos\left(2\frac{\Lambda}{|k|} - t\right)
+ \pi + \sin\left[2 \arccos\left(2\frac{\Lambda}{|k|} - t\right)\right]\right).
\end{align*} 
This immediately shows that 
\[ \lim_{k\to 0} B^0(k) = \frac \pi {16\vF}, 
\label{eq:B_at_zero_free} 
\]
which obviously depends on $\vF$ in contrast to what we will find for $B(k)$ below.
\end{remark}

\medskip

Let us now come back to the function $B(k)$. The following lemma says that the function $B(k)$ vanishes at $k=0$. However, it only vanishes logarithmically.

\begin{lemma}
We have that
\[
{B(k) \underset{k\to0}{\sim} \frac\pi{4\,\log\left(\frac{\Lambda}{|k|}\right)}.} 
\]
\end{lemma}

\begin{proof}
In~\eqref{eq:last_formula_B}, we split the integral in $w$ as follows:
$$\int_{1/2}^{\Lambda/|k|}dw=\int_{1/2}^{\log^2(\Lambda/|k|)}dw+\int_{\log^2(\Lambda/|k|)}^{\Lambda/|k|}dw.$$
By doing so we can write $B(k)=B_1(k)+B_2(k)$ with an obvious definition.
For $w\in [\log^2(\Lambda/|k|),\Lambda/|k|]$ we just use that $g\geq g(1)$ and we obtain
\begin{equation*}
B_2(k)\leq \frac{1}{\pi(\vF+g(1))}
\int_{\log^2(\Lambda/|k|)}^{\Lambda/|k|}dw \int_{0}^{\min\{ 1/2, \Lambda/|k| - w\} }dv\, \frac{  \sqrt{1/4-v^2} }{w\sqrt{ w^2 - 1/4} }.
\end{equation*}
Using the computation of Remark~\ref{rmk:B_no_exchange} we see that the right hand side behaves like
$$\frac 1{8(\vF+g(1))} \left(\frac{\pi}{2}-\arccos\left(\frac{1}{\log^2(\Lambda/|k|)}\right)\right)\underset{|k|\to0}{\sim}\frac{1}{8(\vF+g(1))\log^2(\Lambda/|k|)}.$$
In particular,
$$\lim_{k\to0}\log(\Lambda/|k|) B_2(k)=0.$$
Now for $w\in [1/2,\log^2(\Lambda/|k|)]$, we can safely expand the terms involving $g$. Indeed, we have $w\pm v\leq \log^2(\Lambda/|k|)+1/2$, hence 
$$\frac{\Lambda}{|k|\,(w\pm v)}\geq \frac{\Lambda}{|k|\,(\log^2(\Lambda/|k|)+1/2)}\underset{k\to0}{\longrightarrow}\ii.$$
By the dominated convergence theorem we find that
$$\lim_{k\to0}\log(\Lambda/|k|)B_1(k)= 
\lim_{k\to0}\frac 4\pi\int_{1/2}^{\log^ 2(\Lambda/|k|)}dw \int_{0}^{\min\{ 1/2, \Lambda/|k| - w\} }dv\, \frac{  \sqrt{1/4-v^2} }{w\sqrt{ w^2 - 1/4}}.$$
The right hand side is again similar to $B^0(k)$ computed in Remark~\ref{rmk:B_no_exchange} and it is equal to $\pi/4$. 
\end{proof}

Note that, in contrast to the no-exchange function $B^0$ which has the $\vF$-dependent finite limit $\pi/(16\vF)$ at $k=0$, the true function $B(k)$ tends to zero and its behavior on first order is universal, it does \emph{not} depend on $\vF$. As we have explained, if we neglect the exchange density coming from $R_\lambda$, then we find that in first order in $\lambda$
$$\widehat{\rho_Q}(k)\underset{\lambda\ll1}\simeq \lambda\frac{B(k)}{1+B(k)}\hat\nu(k)\underset{\substack{\lambda\ll1\\ |k|\ll1}}\simeq \lambda\,\frac\pi{4\,\log\left(\frac{\Lambda}{|k|}\right)}\hat\nu(k).$$
That this density vanishes at $k=0$ means that the response of the graphene sheet in the presence of $\nu$ is essentially neutral. This is stark contrast with the massive 3D case, in which there is always (partial) screening, i.e., $B(0)>0$.
However the fact that here $B(k)\to0$ only logarithmically creates some long range oscillations in $x$ space which induce some weak screening effects. Indeed, if we compute the first-order polarization charge of graphene in a ball of radius $R$ (using a smooth localization function $\chi$), we get 
$$\int_{\R^2}\rho_Q(x)\chi(x/R)\,dx\simeq \lambda\int_{\R^2}\frac{\pi\widehat{\nu}(k)}{4\,\log\left(\frac{\Lambda}{|k|}\right)}\widehat\chi(Rk)\,R^2dk\underset{R\to\ii}\sim \lambda\frac{\pi\int_{\R^2}\nu}{4\log(R)}.$$
It goes to zero when $R\to\ii$, but very slowly. So in a finite ball of radius $R$ we might have the impression that $\int_{B_R}\rho_Q>0$ is not small, i.e., that the external density $\nu$ is partially screened by the spontaneaous polarization of graphene.

In this discussion we have neglected the first order density $\rho_{1,X}$ coming from the exchange term $R_\lambda$. The effect of this additional term is not clear to us.

\section{Proof of Theorem~\ref{thm:existence_pol_vac}}\label{sec:proof}

Let $(Q_n)$ be a minimizing sequence for~\eqref{eq:def_min_BDF}. Using the strict inequality~\eqref{eq:cond_v_F_bis_strict}, it is easy to see from our estimates that $(Q_n)$  is uniformly bounded in $\tilde\X$:
\begin{align}
\tr\big(|\cD^0|^{1/2}Q_n^2|\cD^0|^{1/2}\big)& \leq C, \label{unifX} \\
\tr\big(|\cD^0|^{1/2}(Q_n^{++} - Q_n^{--})|\cD^0|^{1/2}\big) & \leq C. \label{unifE}
\end{align}
Therefore, up to a subsequence, there exists an element $Q \in \tilde\X$, such that 
$$ Q_n \rightharpoonup Q \quad \text{weakly-$\ast$ in $\tilde\X$,}$$
similarly as in Remark~\ref{rmk:weak_CV}. This implies in particular that $\rho_{Q_n}\rightharpoonup\rho_Q$ weakly in $\cC$ and strongly in $L^p_{\rm loc}(\R^2)$, by Lemmas~\ref{lem:L1_loc} and~\ref{lem:Coulomb}. Since we also have 
$$\pscal{\phi,Q_n\phi'}\to\pscal{\phi,Q\phi'},\quad \forall \phi,\phi'\in\gH_\Lambda,$$
it is clear that the constraint $-P^0_-\leq Q\leq P^0_+$ is satisfied. Therefore $Q$ is an admissible state for the minimization problem~\eqref{eq:def_min_BDF}.

\bigskip

\noindent {\em Step 1}: Our goal is to show that the energy is lower semi-continuous, i.e. 
\begin{equation}\label{lowsemi}
\liminf_{n\to\ii} \cE^V(Q_n) \geq \cE^V(Q). 
\end{equation}
Observe that the function 
$$ \rho \to D(\rho - \nu,\rho - \nu)$$ is lower semi-continuous.
Therefore $$ \liminf_{n\to\ii} D(\rho_{Q_n} - \nu,\rho_{Q_n} - \nu) \geq D(\rho_Q - \nu, \rho_Q-\nu),$$
using the simple fact that $\rho_{Q_n} \rightharpoonup \rho_{Q}$ weakly in the Coulomb norm $\cC$. 
Since
\begin{multline*}
\tr |\cD^0|(Q^{++}_n - Q^{--}_n) \\ = \tr\big(|\D^0| - |p|(\vF+g(1))\big)(Q^{++}_n - Q^{--}_n) +  (\vF+g(1)) \tr |p|(Q^{++}_n - Q^{--}_n),
\end{multline*}
and $|\D^0| - |p|(\vF+g(1))\geq 0  $, we can use Fatou's Lemma for trace-class operators to obtain
$$ \tr\big(|\D^0| - |p|(\vF+g(1))\big)(Q^{++}_n - Q^{--}_n) \geq \tr\big[|\D^0| - |p|(\vF+g(1))\big](Q^{++} - Q^{--}),$$
and \eqref{lowsemi} will be achieved by showing 
\begin{multline}\label{semil}
\liminf_{n}\left( \tr |p|(Q^{++}_n - Q^{--}_n) - \frac {1}{2(\vF+g(1))} \iint_{\R^2\times \R^2} \frac {|Q_n(x,y)|^2}{|x-y|} dx\,dy\right) \\ \geq \tr |p|(Q^{++} - Q^{--}) - \frac {1}{2(\vF+g(1))}  \iint_{\R^2\times \R^2} \frac {|Q(x,y)|^2}{|x-y|} dx\,dy.
\end{multline}
This inequality will now be proved in {\em Step 2}.

\bigskip

\noindent{\em Step 2:} The proof of \eqref{semil} will heavily rely on \eqref{unifE}, i.e., the uniform boundedness of $Q_n$ in the energy norm $\cD^0$. 
The first ingredient of the proof is to split the space $\R^2$ into a region close to the defect $\nu$, and another one far away, and to show that the energy essentially localizes. A similar idea was used in~\cite{HaiLewSer-05b} but because of the absence of a mass, we have to argue differently.

Let us consider two real functions $\chi, \eta \in C^\infty([0,\infty);[0,1])$, such that $\chi = 1$ on $[0,1]$ and $\chi=0$ on $[2,\infty)$ and $\chi^2 + \eta^2 = 1$. 
We define $$ \chi_R(x) = \chi(|x|/R), \,\,\, \eta_R(x) = \eta(|x|/R)\,\,\, \forall x \in \R^2.$$ 
Now we apply the localization Lemma \ref{localization} given in the Appendix (and which is a consequence of the localization estimate obtained in \cite{LenLew-10}) in order to obtain
\begin{multline}
\tr |p| ( Q_n ^{++} - Q_n^{--}) \geq \tr |p| \chi_R (Q_n ^{++} - Q_n^{--}) \chi_R + \tr |p| \eta_R (Q_n ^{++} - Q_n^{--}) \eta_R \\
- c \left( \tr |p|  (Q_n ^{++} - Q_n^{--})   \right)^{1/2} \left( \|\nabla \chi_R\|_{L^2}^2 + \|\nabla \eta_R\|_{L^2}^2 \right)^{1/2}\left(\|\nabla \chi_R\|^2_{L^4} + \|\nabla \eta_R \|_{L^4}^2 \right)^{1/2}.
\end{multline}
Since
\begin{align*} 
\|\nabla \chi_R\|_{L^2}^2  = \|\nabla \chi\|^2_{L^2},\quad \| \nabla \chi_R\|^2_{L^4}  = \frac 1R \| \nabla \chi \|^2_{L^4},
\end{align*}
and analogously for $\eta_R$, the uniform boundedness of $\tr |p| ( Q_n ^{++} - Q_n^{--})$, implies
\[
\tr |p| ( Q_n ^{++} - Q_n^{--}) \geq \tr |p| \chi_R (Q_n ^{++} - Q_n^{--}) \chi_R + \tr |p| \eta_R (Q_n ^{++} - Q_n^{--}) \eta_R - \frac{C}{R},
\]
for an $n$-independent constant $C$.
On the other hand we can write 
\begin{multline*}
 \iint_{\R^2\times\R^2} \frac{|Q_n(x,y)|^2}{|x-y|} dx\,dy =  \iint_{\R^2\times\R^2} \frac{\chi_R^2(x) |Q_n(x,y)|^2}{|x-y|} dxdy\\ +   \iint_{\R^2\times\R^2} \frac{\eta_R^2(x) |Q_n(x,y)|^2}{|x-y|} dx\,dy. 
\end{multline*}
Using now $$ Q_n^{++} - Q_n^{--} \geq Q_n^2,$$
as well as Kato's inequality~\eqref{eq:Kato}, we obtain, similar to \cite{HaiLewSer-05b}, that under the assumptions on $\vF$
$$  \tr |p| \eta_R [Q_n ^{++} - Q_n^{--}] \eta_R \geq  \tr |p| \eta_R Q_n^2 \eta_R \geq \frac 1{ \vF+g(1)} \int\int \frac{\eta_R^2(x) |Q_n(x,y)|^2}{|x-y|} dxdy.$$
In order to conclude {\em Step 2}, it remains to show that 
\begin{multline}
\lim_{R\to \infty} \lim_{n\to \infty} \left(  \tr |p|\chi_R (Q^{++}_n - Q^{--}_n)\chi_R  - \frac {1}{2(\vF+g(1))} \int \frac {\chi_R^2 (x) |Q_n(x,y)|^2}{|x-y|} dxdy\right) \\ = \tr |p|(Q^{++} - Q^{--}) - \frac {1}{2(\vF+g(1))} \int \frac {|Q(x,y)|^2}{|x-y|} dxdy,
\label{eq:to_be_shown}
\end{multline}
where it is important to do the limit $n \to \infty$ first and then the limit $R \to \infty$. 
To this end we consider the two terms  
$$  \tr |p|\chi_R (Q^{++}_n - Q^{--}_n)\chi_R\quad \text{and}\quad \int \frac {\chi_R^2 (x) |Q_n(x,y)|^2}{|x-y|} dxdy,$$ separately. 
First, we observe that, thanks to the cut-off 
\begin{equation}
 \tr|p|\chi_R (Q^{++}_n - Q^{--}_n)\chi_R\\ = \tr|p| \chi_R \Pi_\Lambda (Q^{++}_n - Q^{--}_n)\Pi_\Lambda \chi_R. 
\end{equation}
As we have already seen, the operator $\Pi_\Lambda\chi_R$ is Hilbert-Schmidt. On the other hand
\begin{align*}
|p| \chi_R \Pi_\Lambda&=\frac{p}{|p|}\cdot\Big( [p,\chi_R]\Pi_\Lambda+ \chi_R p\Pi_\Lambda\Big)\\
&= \frac{p}{|p|}\cdot \Big(-i(\nabla\chi_R)\Pi_\Lambda+ \chi_R \,p\,\Pi_\Lambda\Big)
\end{align*}
is Hilbert-Schmidt since $p|p|^{-1}$ is bounded and $(\nabla\chi_R)\Pi_\Lambda$ and $\chi_R \,p\,\Pi_\Lambda$ are both Hilbert-Schmidt. Now we can use that $Q_n\rightharpoonup Q$ weakly-$\ast$ in $\cB(\gH_\Lambda)$ which implies in particular that $\tr(BQ_nB')\to\tr(BQB')$ for all $B,B'\in\gS^2$, and we obtain
\[
\lim_{n\to\ii} \tr|p|\chi_R (Q^{++}_n - Q^{--}_n)\chi_R = \tr|p|\chi_R (Q^{++} - Q^{--})\chi_R. 
\]
Passing then to the limit $R\to\ii$ gives
$$ \lim_{R\to \infty} \lim_{n\to \infty}   \tr |p|\chi_R (Q^{++}_n - Q^{--}_n)\chi_R =  \tr |p| (Q^{++} - Q^{--}). $$

It remains to prove the convergence of the exchange term in~\eqref{eq:to_be_shown}. Recall $$A_n := Q_n  - Q \rightharpoonup 0 \quad\text{in $\tilde\X$}$$ and that 
$ \tr|\cD^0| Q_n^2$ is uniformly bounded.
Notice first that, by the Cauchy-Schwarz inequality, it suffices to show 
\[
\lim_{R\to\ii}\lim_{n\to\ii} \iint_{\R^2\times\R^2} \frac {\chi_R^2 (x) |A_n(x,y)|^2}{|x-y|} dx\,dy = 0
\]
to conclude the second step of the proof. 
To this end we introduce an additional decomposition in the infrared regime in the following form:
\begin{align*}
A_n & \, = \big(\1(|p|\leq\eps)+\1(|p|\geq\eps)\big)A_n \big(\1(|p|\leq\eps)+\1(|p|\geq\eps)\big) \\
& \, = \1(|p|\geq\eps) A_n \1(|p|\geq\eps) + B_n,
\end{align*}
such that the first term on the right hand side
$$ \1(|p|\geq\eps) A_n \1(|p|\geq\eps) :=A_n^\eps$$
consists only of momenta larger than a given $\eps>0$. 
Since by assumption $A_n |p|^{1/2}\wto0$ weakly in $\gS^2$, we see that $A_n^\eps \wto0$ in $\gS^2$. Due to the cut-off in Fourier space, we deduce that the kernel $A_n^\epsilon(x,y)$ of $A_n^\eps$ converges to zero weakly in $H^{s}(\R^2\times\R^2)$ for all $s\geq0$, and hence strongly in $L^p_{\rm loc}(\R^2\times \R^2)$ for all $p$. Now we simply decompose
\begin{multline*}
 \iint_{\R^2\times\R^2} \frac {\chi_R^2 (x) |A_n^\eps(x,y)|^2}{|x-y|} dx\,dy= \iint_{\R^2\times\R^2}\frac {\chi_R^2 (x)\chi_{5R}^2(y) |A_n^\eps(x,y)|^2}{|x-y|} dx\,dy\\
+ \iint_{\R^2\times\R^2} \frac {\chi_R^2 (x)\eta_{5R}^2(y) |A_n^\eps(x,y)|^2}{|x-y|} dx\,dy. 
\end{multline*}
The first term on the right hand side converges to zero since  $A_n^\epsilon(x,y)\to0$ strongly in, say, $L^6(B_{2R}\times B_{10R})$ whereas $\chi_R^2 (x)\chi_{5R}^2(y)|x-y|^{-1}$ belongs to $L^{3/2}(\R^2\times\R^2)$. The second term on the right hand side is bounded by
$$ \iint_{\R^2\times\R^2}\frac {\chi_R^2 (x)\eta_{5R}^2(y) |A_n^\eps(x,y)|^2}{|x-y|} dx\,dy\leq \frac{C}{R}\tr(A_n^\eps)^2\leq \frac{C}{R\eps}\tr|p|(A_n)^2\leq \frac{C}{R\eps}.$$

It remains to consider the terms $$B_n = B_n^1 + B_n^2 + B_n^3$$
depending on the location of the cut-off function $\1(|p|\leq\eps)$ with the only important point being that either of the terms has one such IR-term $\1(|p|\leq\eps)$  on at least one side.
Again with Kato's inequality~\eqref{eq:Kato} we bound for $k=1,2,3$
\begin{align*}
 \iint_{\R^2\times\R^2} \frac {\chi_R^2 (x) |B^k_n(x,y)|^2}{|x-y|} dx\,dy &\leq  \iint_{\R^2\times\R^2} \frac { |B^k_n(x,y)|^2}{|x-y|} dx\,dy\\
 &\leq C\min\Big(\tr (|p| (B^k_n)^* B^k_n)\;,\;\tr (|p| B^k_n (B^k_n)^*)\Big). 
\end{align*}
Choosing on the right the term for which $\1(|p|\leq\eps)$ hits $|p|$, we see that 
$$ \iint_{\R^2\times\R^2} \frac {\chi_R^2 (x) |B^k_n(x,y)|^2}{|x-y|} dx\,dy\leq C\tr (|p|\1(|p|\leq\eps) (A_n)^2)$$
for $k=1,2,3$. At this point we use the infrared logarithmic divergence of the effective Fermi velocity to estimate
$$\tr \big(|p|\1(|p|\leq\eps) (A_n)^2\big)\leq \frac{C}{\log(\Lambda/\eps)}\tr|\cD^0|(A_n)^2\leq \frac{C}{\log(\Lambda/\eps)}.$$
If we take first $n\to\ii$, then $R\to\ii$ and finally $\eps\to0$, we have concluded {\em Step 2} of the proof. 

To summarize, we have shown that for a minimizing sequence $(Q_n)$,
$$E_{\rm BDF}^V=\liminf_{n\to\ii}\cE^V_{\rm BDF}(Q_n)\geq \cE^V_{\rm BDF}(Q)\geq E_{\rm BDF}^V$$
where we recall $E_{\rm BDF}^V$ is the infimum of the BDF energy. So we deduce that $\cE^V_{\rm BDF}(Q)=E_{\rm BDF}^V$ and that $Q$ is a minimizer.

\bigskip

\noindent{\em Step 3}: It remains to prove that it satisfies the self-consistent equation. This can be done 
as in \cite[Lemma 2]{HaiLewSer-05b} with the additional problem that there is no gap in the spectrum of the mean-field operator. For completeness let us indicate the idea of the proof. 
Since $Q$ is a minimizer, then for any other admissible state $Q'$, we know by convexity of the constraint that 
$$ \frac{d}{dt} \cE^V((1-t)Q + tQ')\Big |_{t=0} \geq 0,$$  which implies that
\begin{equation}
\tr \cD^0 (Q'-Q) +D(\rho_Q-\nu,\rho_{Q'-Q})+\Re\int_{\R^2}\int_{\R^2}\frac{\overline{Q(x,y)}(Q'-Q)(x,y)}{|x-y|}\,dx\,dy\geq0
\label{eq:positive}
\end{equation}
Let us remark that the operator 
$$\cD_Q:=\cD^0+(\rho_Q-\nu)\ast\frac{1}{|\cdot|}-\frac{Q(x,y)}{|x-y|}$$
is self-adjoint on the same domain as $\cD^0$. This follows from Rellich's theorem, since the two self-consistent terms are relatively bounded with respect to $\cD^0$, with relative bound as small as we want. For instance, we have for any $\phi\in\gH_\Lambda$
\begin{align*}
\norm{(\rho_Q-\nu)\ast\frac{1}{|\cdot|}\phi}_{L^2(\R^2)}&\leq \norm{(\rho_Q-\nu)\ast\frac{1}{|\cdot|}}_{L^4(\R^2)}\norm{\phi}_{L^4(\R^2)}\\
&\leq C\norm{\rho_Q-\nu}_{\cC}\norm{\phi}_{L^4(\R^2)}\\
&\leq C\norm{\rho_Q-\nu}_{\cC}\norm{|p|^{1/2}\phi}_{L^2}\\
&\leq C\norm{\rho_Q-\nu}_{\cC}\left(\eps\norm{\cD^0\phi}+\frac{1}{2\eps(\vF+g(1))}\norm{\phi}\right)
\end{align*}
where we have used the Sobolev inequality $\norm{f}_{L^4(\R^2)}\leq C\norm{|p|^{1/2}f}_{L^2(\R^2)}$ and
$$\norm{\rho\ast\frac{1}{|\cdot|}}_{L^4(\R^2)}\leq C\norm{|p|^{1/2}\rho\ast\frac{1}{|\cdot|}}_{L^2(\R^2)}=CD(\rho,\rho)^{1/2}.$$
The argument is similar for the exchange term: 
\begin{align*}
\left|\int_{\R^2}\frac{Q(x,y)}{|x-y|}\phi(y)\,dy\right|&\leq \left(\int_{\R^2}\frac{|Q(x,y)|^2}{|x-y|}\,dy\right)^{1/2}
\left(\int_{\R^2}\frac{|\phi(y)|^2}{|x-y|}\,dy\right)^{1/2}\\
&\leq C\left(\int_{\R^2}\frac{|Q(x,y)|^2}{|x-y|}\,dy\right)^{1/2}\pscal{\phi,|p|\phi}^{1/2}\\
&\leq C\left(\int_{\R^2}\frac{|Q(x,y)|^2}{|x-y|}\,dy\right)^{1/2}\left(\eps\norm{\cD^0\phi}+\frac{1}{2\eps(\vF+g(1))}\norm{\phi}\right).
\end{align*}
Taking the square and integrating with respect to $x$ gives the infinitesimal relative boundedness.

Now we choose 
\begin{equation}\label{Qprime}
Q':=\1_{(-\ii,0)}(\cD_Q)-P^0_-,
\end{equation}
and we claim that
\begin{align}
&\tr_0 \cD^0 (Q'-Q) +D(\rho_Q-\nu,\rho_{Q'-Q})+\Re \iint_{\R^2\times\R^2}\frac{\overline{Q(x,y)}(Q'-Q)(x,y)}{|x-y|}\,dx\,dy\nonumber
\\
&\qquad = \tr|\cD_Q|\big(P'(Q'-Q)P'-(P')^\perp(Q'-Q)(P')^\perp\big)\label{eq:first-line}\\
&\qquad \leq -\tr|\cD_Q|(Q'-Q)^2\label{eq:negative}
\end{align}
The first line~\eqref{eq:first-line} would just be the linearity of the trace if all the operators were trace-class. Because of our generalized definition of the trace, it is more complicated to verify~\eqref{eq:first-line}. With a gap this was done in~\cite{HaiLewSer-05a} and without a gap similar arguments have been used in~\cite{FraHaiSeiSol-12} and in~\cite{FraLewLieSei-11b}. We will not discuss this point further. Putting~\eqref{eq:positive} and~\eqref{eq:negative} together, we reach the conclusion that $Q=Q'$ except possibly on the kernel of $\cD_Q$. Therefore
$$Q+P^0_-=\1_{(-\ii,0)}(\cD_Q)+\delta$$
for some $0\leq\delta\leq \1_{\{0\}}(\cD_Q)$. If $\ker(\cD_Q)=\{0\}$ then of course $\delta=0$. 

If $\ker(\cD_Q)\neq\{0\}$ but has dimension $\geq2$, it is a well-known fact that $\delta=0$ or $\delta=\1_{\{0\}}(\cD_Q)$ that is, $\delta$ must fill the last shell completely, see~\cite{BacLieLosSol-94,BacLieSol-94} and~\cite[Prop. 3]{HaiLewSer-09}. If $\dim\ker(\cD_Q)=1$, and $Q+P_0$ is a projector, then, necessarily, $\delta = 0$ or $\delta =  \1_{\{0\}}(\cD_Q)$ and we are done. If, however, $Q + P_0$ is not a projector, then the argument does not work but, in this case, the energy does not change if we subtract $\delta$ from $Q$, because the corresponding particle does not interact with itself, i.e., $Q'=Q-\delta$ is a minimizer as well. Therefore we can redo the above argument with $Q'$ instead of $Q$. But now, thanks to our definition \eqref{Qprime}, we actually know that $Q' + P_0$ is a projector.  Hence, we deduce, as stated, that $Q'=\1_\mathcal{I}(\cD_{Q'})-P^0_-$ with $\mathcal{I}=(-\ii,0)$ or $\mathcal{I}=(-\ii,0]$. \qed

\begin{remark}
Since $Q$ must be a minimizer for~\eqref{eq:def_min_BDF}, we deduce that 
$$\lim_{n\to\ii}\cE^V_{\rm BDF}(Q_n)= \cE^V_{\rm BDF}(Q).$$ 
More precisely, we see that 
$$\lim_{n\to\ii}D(\rho_{Q_n},\rho_{Q_n})=D(\rho_{Q},\rho_{Q})$$
which implies that $\rho_{Q_n}\to\rho_Q$ strongly in $\cC$. Similarly we have 
$$\lim_{n\to\ii}\tr|\D^0|(Q^{++}_n - Q^{--}_n) = \tr|\D^0|(Q^{++} - Q^{--})$$
and
$$\lim_{n\to\ii} \iint_{\R^2\times\R^2} \frac{|Q_n(x,y)|^2}{|x-y|}dx\,dy=\lim_{n\to\ii} \iint_{\R^2\times\R^2} \frac{|Q(x,y)|^2}{|x-y|}dx\,dy.$$
By the reciprocal of Fatou's lemma for operators, this proves strong convergence of $|\cD^0|^{1/2}Q_n^{\pm\pm}|\cD^0|^{1/2}$ in $\gS^1$, and strong convergence of $Q_n(x,y)|x-y|^{-1/2}$ in $\gS^2$.
\end{remark}

\appendix
\section{Localization of massless kinetic energy}

In the following lemma, we provide an IMS-type formula for the massless relativistic energy of a fermionic density matrix $0\leq\gamma\leq 1$, based on~\cite{LenLew-10}. Results of the same form already exist in the literature, see in particular~\cite[Thm. 9]{LieYau-88}. There, a simple explicit formula for the operator $|p| -\chi|p| \chi- \eta|p| \eta$ is used.

\begin{lemma}\label{localization}
Let $0\leq \gamma \leq 1$ such that $\tr|p|\gamma < \infty.$ Then for a partition of unity $\chi^2 + \eta^2 =1$, with smooth functions $\chi, \eta $ 
one obtains in dimension $d=2$
\begin{multline*}
\tr |p| \gamma \geq \tr |p| \chi \gamma \chi + \tr |p| \eta \gamma \eta  \\
- c \left( \tr |p|  \gamma  \right)^{1/2} \left( \|\nabla \chi\|_{L^2}^2 + \|\nabla \eta\|_{L^2}^2 \right)^{1/2}\left(\|\nabla \chi\|^2_{L^4} + \|\nabla \eta \|_{L^4}^2 \right)^{1/2},
\end{multline*}
and in dimension $d=3$ 
\begin{multline*}
\tr |p| \gamma \geq \tr |p| \chi \gamma \chi + \tr |p| \eta \gamma \eta  \\
- c \left( \tr |p|  \gamma  \right)^{2/3} \left( \|\nabla \chi \|_{L^2}^2 + \|\nabla \eta\|_{L^2}^2 \right)^{1/3}\left(\|\nabla \chi\|^2_{L^6} + \|\nabla \eta \|_{L^6}^2 \right)^{2/3},
\end{multline*}
where $c$ is a universal constant.
\end{lemma}
\begin{proof}
We use \cite[Lemma A.1]{LenLew-10} which states that
$$
|p|  - \chi |p| \chi - \eta |p| \eta  \geq 
- \frac 1{\pi} \int_0^\infty \sqrt{t}\, dt\, \frac 1{t+ p^2} \left( |\nabla \chi |^2 + |\nabla \eta|^2 \right)\frac 1{t+ p^2} .
$$
So, it remains to estimate the term 
$$ \tr \left(  \int_0^\infty \sqrt{t}\, dt\, \frac 1{t+ p^2} \left( |\nabla \chi |^2 + |\nabla \eta|^2 \right)\frac 1{t+ p^2}  \gamma
 \right).  $$
 We are going to decompose the integral into a small $t$ and large $t$-part. 
For that reason we will use two different estimates for the integrand. 
Define 
$$ f^2(x) =  |\nabla \chi (x)|^2 + |\nabla \eta(x)|^2,$$ 
Since $0\leq \gamma \leq 1$ we obtain 
\begin{multline}\label{est1}
\tr \frac 1{t+ p^2} f^2\frac 1{t+ p^2}  \gamma \leq \tr \frac 1{t+ p^2} f^2 \frac 1{t+ p^2}  = \int f^2(x) dx \frac 1{(2\pi)^{d/2}} \int \frac {dp}{(t+p^2)^2}  \\ = t^{d/2 - 2} \int f^2(x) dx \int \frac {dp}{(1+p^2)^2} ,
\end{multline}
on the other hand 
\begin{multline}\label{est2}
\tr \frac 1{t+ p^2} f^2\frac 1{t+ p^2}  \gamma = \tr |p|^{-1/2}  \frac 1{t+ p^2} f^2 \frac 1{t+ p^2} |p|^{-1/2} |p| \gamma \\ \leq \left\| \frac 1{(t + p^2) |p|^{1/2}} f^2 \frac 1{(t + p^2) |p|^{1/2}}  \right\|\tr |p| \gamma 
\leq \frac 1{t^2} \left \| \frac 1{|p|^{1/2}} f \right \|^2 \tr |p| \gamma.
\end{multline}
Recall that here $f$ is seen as a multiplication operator on $L^2(\R^2)$. In dimension $d=2$ the Hardy-Littlewood-Sobolev inequality states 
\begin{align*}
\left \| \frac 1{|p|^{1/2}} f\right\|_{L^2(\R^2)\to L^2(\R^2)}^2&= \sup_{\norm{\phi}_{L^2(\R^2)}=1}\left \| (-\Delta)^{-1/4} (f\phi)\right\|_{L^2(\R^2)}^2\\
&=\sup_{\norm{\phi}_{L^2(\R^2)}=1}\int_{\R^2}\int_{\R^2}  \frac{\overline{f(x)\phi(x)}f(y)\phi(y)}{|x-y|}dx\,dy\\
&  \leq C\sup_{\norm{\phi}_{L^2(\R^2)}=1}\| f\phi\|^2_{L^{4/3}(\R^2)}=C\norm{f}_{L^4(\R^2)}^2.
\end{align*}
Combining these results we get in $d=2$ 
\begin{align*} 
\tr \left(  \int_0^\infty \sqrt{t}\, dt\, \frac 1{t+ p^2} f^2 \frac 1{t+ p^2}  \gamma \right) 
&\leq  C\left(\int_0^\eps \frac{dt}{\sqrt{t}} \|f\|_{L^2(\R^2)}^2 + \int_{\eps}^{\infty} \frac{dt}{t^{3/2}}  \| f\|^2_{L^4(\R^2)} \tr |p| \gamma\right)\\
& \leq C\left(\sqrt{\eps}  \|f\|_{L^2(\R^2)}^2  + C\frac 1{\sqrt {\eps}}  \| f\|^2_{L^4(\R^2)} \tr |p| \gamma \right)\\ 
 &\leq C\|f\|_{L^2(\R^2)} \|f\|_{L^4(\R^2)} (\tr |p| \gamma)^{1/2},
\end{align*}
after optimizing over $\eps$, which proves the Lemma in 2D. For $d=3$ one proceeds similar with the corresponding Hardy-Littlewood-Sobolev inequality in 3D, leading to
$$\left \| \frac 1{|p|^{1/2}} f\right\|_{L^2(\R^3)\to L^2(\R^3)}^2\leq C\norm{f}_{L^6(\R^3)}^2.$$
This concludes the proof of Lemma~\ref{localization}.
\end{proof}


\end{document}